\documentclass[12pt, letterpaper, table]{article}
\usepackage[utf8]{inputenc}
\usepackage[dvipsnames]{xcolor}
\usepackage{bbold}
\usepackage{colonequals}
\usepackage{amsmath, amsthm, bm, bbm, amssymb}
\usepackage{natbib}
\usepackage{pifont}
\usepackage[margin=1in]{geometry}
\usepackage{mathtools}
\usepackage{bigints}
\usepackage[font={footnotesize}]{caption,subcaption}
\usepackage{placeins} 
\usepackage{hyperref}
\usepackage{dutchcal}
\usepackage{enumitem}

\makeatletter
\newcommand{\sDelta}{{\vphantom{\Delta}\mathpalette\sD@lta\relax}}
\newcommand{\sD@lta}[2]{%
  \ooalign{\hidewidth$\m@th#1\mkern-1mu {\scriptstyle s}$\hidewidth\cr$\m@th#1\Delta$\cr}%
}
\makeatother

\makeatletter
\newcommand{\tDelta}{{\vphantom{\Delta}\mathpalette\tD@lta\relax}}
\newcommand{\tD@lta}[2]{%
  \ooalign{\hidewidth$\m@th#1\mkern-1mu {\scriptstyle t}$\hidewidth\cr$\m@th#1\Delta$\cr}%
}
\makeatother

\newcounter{example}
\newenvironment{example}[1][]{\refstepcounter{example}\par\medskip
   \noindent \textbf{Example~\theexample. #1} \rmfamily}{\medskip}

\DeclarePairedDelimiter\ceil{\lceil}{\rceil}
\DeclarePairedDelimiter\floor{\lfloor}{\rfloor}

\newtheorem{definition}{Definition}

\newcommand{\bs}{\bm{s}}
\newcommand{\bc}{\bm{c}}
\newcommand{\bx}{\bm{x}}
\newcommand{\bfzero}{\bm{0}}
\newcommand{\bI}{\bm{I}}
\newcommand{\bA}{\bm{A}}
\newcommand{\bZ}{\bm{Z}}
\newcommand{\bQ}{\bm{Q}}
\newcommand{\bM}{\bm{M}}

\newcommand{\bv}{\bm{v}}
\newcommand{\bz}{\bm{z}}
\newcommand{\bfepsilon}{\bm{\epsilon}}
\newcommand{\bfpsi}{\bm{\psi}}
\newcommand{\bfxi}{\bm{\xi}}

\newcommand{\bftheta}{\bm{\theta}}

\newcommand{\normal}{\mathcal{N}}

\newtheorem{fact}{Fact}
\newtheorem{proposition}{Proposition}
\newtheorem{assumption}{Assumption}

\title{Statistical Inference for Complete and Incomplete Mobility Trajectories under the Flight-Pause Model}

\author{Marcin Jurek\thanks{Corresponding author: \texttt{marcinjurek1988@gmail.com}} \qquad Catherine A. Calder \qquad Corwin Zigler\\\vspace{0.01cm}\\Department of Statistics and Data Sciences, University of Texas at Austin}

\begin{document}

\maketitle

\begin{abstract}
We formulate a statistical flight-pause model for human mobility, represented by a collection of random objects, called motions, appropriate for mobile phone tracking (MPT) data. We develop the statistical machinery for parameter inference and trajectory imputation under various forms of missing data. We show that common assumptions about the missing data mechanism for MPT are not valid for the mechanism governing the random motions underlying the flight-pause model, representing an understudied missing data phenomenon.  We demonstrate the consequences of missing data and our proposed adjustments in both simulations and real data, outlining implications for MPT data collection and design. 
\end{abstract}


\section{Introduction}\label{sec:introduction}

Over the past decade, smartphones equipped with location-sensing technologies – such as multilateration of radio signals between cell towers, global navigation satellite systems, or connection to Wi-Fi positioning systems -- have become ubiquitous throughout much of the world \citep{pewStudy}. These mobile-phone tracking (MPT) technologies, used individually or in concert, supply precise geographic information to smartphone applications for purposes such as real-time navigation, locating network partners (e.g., Find My Friends), or recording fitness achievements. They also provide a wealth of data relevant to the study of daily patterns of human mobility, both about individual behavior and from a systems perspective.   In the biomedical and social sciences, MPT data is becoming an increasingly common component of cohort studies, where it has been employed for purposes of digital phenotyping \citep{onnela2016harnessing} or estimating personal exposure to the ambient environment or particular social contexts \citep{browning2021human, cagney2020urban, nyhan2019quantifying, schultes2021covid, crawford2021impact}

Statistical analysis of general trajectory data, defined as the spatial location of an object over time, has a rich history \citep{dunn1977analysis, blackwell1997random, brillinger2010modeling}. For example, mechanistically-motivated statistical models for animal movement trajectories, including those with dynamics derived from differential equations, have received considerable attention in recent years \citep[e.g.][]{brillinger2004stochastic, hooten2017animal, russell2018spatially}. In contrast, statistical treatment of MPT trajectories has emphasized a particularly salient feature of daily human mobility – distinct periods of stationarity and of movement – that is often highly relevant to the objectives of biomedical and social scientific investigations. For example, researchers interested in characterizing the response to a novel physical therapy protocol may be interested in the duration or frequency of stationary periods or the distance traveled during periods of movement. Similarly, social scientists may seek to understand the consequences of time spent in places lacking informal supervision on the behavioral outcomes of youth\footnote{We note that the connections to the descriptive summaries of daily human mobility “activity space” and a “space-time prism,” both initially developed in the Geography literature \citep{golledge1997spatial, torsten1970people}. In statistics, the former has been extended by \citep{chen2019generalized} using tools from topological data analysis. }.

Despite the growing interest in MPT data and quantification of some of its features \citep[e.g.][]{chen2020measuring}, rigorous statistical tools to study them remain somewhat limited.  Following \citet{rhee2011levy}’s observation that human walks share characteristic features of truncated Levy random walks, \citet{barnett2020inferring} and \citet{liu2021bidirectional} propose modeling MPT data as a series of flights and pauses.  Motivated by the task of reconstructing portions of MPT data trajectories that are unobserved (i.e., contain “gaps”), they propose a nonparametric approach using observed sequences of flights and pauses at different points in an individual’s trajectory to fill in gaps.  Absent from this work is any formal statement of a generative model for a partially-observed trajectory (e.g., with a likelihood), making it difficult to rigorously investigate the implications of the gaps (e.g., ignorable or non-ignorable missingness) on inferred trajectories. 

Our work shares shares some similarities with the existing contributions to the animal movement literature. Developed independently within that field and anchored in the theory of stochastic processes, the moving-resting (MR) process \citet{yan2014moving, hu2021moving} is similar to the approach described above. It assumes continuous time and specific parametric distributions to capture the characteristics of flights and pauses. The continuous time formulation and attendant complications for inference are motivated by a data collection mechanism that records the position of an animal at irregular time intervals.  In contrast, the present work is grown from literature on human mobility as mentioned above, which is typically based on regularly-spaced (in time) observations, with many observations missing by design. Another paper in the animal movement literature related to ours, \citet{langrock2012flexible}, presents a hidden Markov model for such regularly-spaced observations, but the model formulation is not explicitly based on flights and pauses nor does the paper investigate the implications of missing data patterns. These implications are often fundamental in human mobility studies as discussed in \citet{barnett2020inferring}.

In this paper, we build on the previous work of human mobility researchers and introduce a statistical framework for modeling MPT trajectories based on what we refer to as the \textit{flight-pause model} (FPM), which is discrete in time and continuous in space. At first glance, this model may seem to be an overly simplistic description of daily human mobility.  We argue, however, that it (1) is a sufficiently rich baseline model upon which to build and (2) allows us to demonstrate how likelihood-based inference on model parameters and imputation of gaps can be done under different assumptions about the data collection mechanism or, equivalently, the missing data mechanism.  We view the latter as a particularly important contribution of our work because of the insights that can be gleaned about the implications of MPT data collection strategies in cohort studies.  For example, we can formally characterize the implications of on-off designs – purposeful breaks in data collection to conserve battery power – on bias in parameter estimation and reduction in efficiency.

This paper is organized as follows. In Section \ref{sec:model}, we formally introduce the FPM and provide expressions for the corresponding likelihood function, generally and under a particular parametrization of components of the model. We extend our discussion of likelihood-based inference for the flight-pause model in Section \ref{sec:data-model} by considering the incomplete data setting, describe different data collection strategies yielding incomplete MPT trajectories and their implications in Section \ref{sec:data-collection}, and demonstrate how trajectory interpolation can be performed in Section \ref{sec:imputation}. Sections \ref{sec:simulations} and \ref{subsec:data-application} present numerical simulations and an analysis of real MPT data.  Finally, in Section \ref{sec:discussion}, we conclude with a general discussion of our results. The code used in simulations and data illustration can be found at \href{https://github.com/marcinjurek/pyhMob}{\texttt{https://github.com/marcinjurek/pyhMob}}.

\section{Statistical formulation of the flight pause model}\label{sec:model}

In this section, we introduce the FPM, a probabilistic model for human movement specified as a probability distribution on a random object called a \textit{motion}.  If we assume this distribution is indexed by a collection of unknown parameters, it can be viewed as a statistical model that can be used to make inference on the unknown parameters governing motion, given observed instances of human mobility.  The dynamics of the FPM are guided by ideas introduced in seminal works of \citet{barnett2020inferring} and \citet{rhee2011levy}. In particular, we are motivated by the setting where MPT data come in the form of a sequence of measurements observed at discrete time increments (e.g., every 15 seconds) over a period of time on the order of days, where a measurement may be taken while a person is moving or is stationary.

To introduce the FPM, we define a \textit{motion} $\mathcal{M} = \left(\mathbf{M}_1, \dots, \mathbf{M}_K\right)$ to be a sequence of $K$ random objects called \textit{increments}.\footnote{In keeping with the convention used in mathematics throughout the paper we use brackets when the elements are ordered and curly braces when they are not.}  Each increment, $\mathbf{M}_k$, corresponds to a stage of the motion and can have one of two \textit{types}: \textit{flight} or  \textit{pause}.  Flights are increments that last for one unit of time and represent movement in physical space (here, assumed to be $\mathbb{R}^2$, without loss of generality).  Pauses correspond to periods of stationarity in physical space which last one or more time units. 

Each increment can be decomposed as $\mathbf{M}_k = \left(\mathbf{L}_k, \mathbf{\Delta}_k, R_k\right)$.  The first component of the increment, $\mathbf{L}_k$, consists of information describing its beginning:  $\mathbf{L}_k = (L^T_k, L^S_k)$, where $L^T_k \in \mathbb{N}$ and $L^S_k \in \mathbb{R}^2$ denote the time and location, respectively, at which the $k$-th increment starts. Notice that we assume time is discrete and can be mapped onto the integers without loss of generality.  The second and third component of an increment, $\mathbf{\Delta}_k$ and $R_k$, capture what happens during an increment.  We let  $\mathbf{\Delta}_k = (\Delta^T_k, \Delta^S_k)$, where $\Delta^T_k \in \mathbb{N}$, stands for how long the increment lasts, and $\Delta^S_k \in \mathbb{R}^2$ is a vector in $\mathbb{R}^2$ representing displacement in physical space.  Finally, $R_k \in \{f, p\}$ is an indicator of whether the increment is a flight ($R_k = f$) or pause ($R_k=p$).  

We now define the FPM through a collection of restrictions imposed on the density of a motion, $q(\mathcal{M})$.  Specifically, under the FPM, we will subsequently show that:
\begin{equation}\begin{split}\label{eq:FPM-basic}
q(\mathcal{M}) = q(\bM_1) &\cdot \prod_{k=2}^K q(R_k|\bM_{k-1}) \cdot \prod_{R_k = f, k>1} q(\Delta^S_k|R_k = f, \bM_{k-1}) \\ 
&\cdot \prod_{R_k = p, k>1} q(\Delta^T_k|R_k = p, \bM_{k-1}).
\end{split}
\end{equation}

The exact derivation of \eqref{eq:FPM-basic}, contained within a proof of Proposition \ref{thm:complete-data-likelihood}, are shown in the Appendix.

We start with a set of properties which ensures continuity of the motion. Continuity should not be understood here in the strict sense in which it is defined in calculus. Instead we use this term to describe formally what it takes for an increment to start at the location in space and time right after the end of the previous one. Consequently, we assume that for $k=2, \dots, K$ the following equalities hold almost surely.
\begin{align}
  L^T_k &= L_{k-1}^T + \Delta^T_{k-1}  &  &\quad\quad\text{(continuity in time)} \label{ass:continuity-in-time} \\
  L^S_k &= L^S_{k-1} + \Delta^S_{k-1}, & \text{if } R_{k-1} = f  &\quad\quad\text{(continuity in space - flights)} \label{ass:continuity-in-space-flights} \\
  L^S_k &= L^S_{k-1}, &\text{if } R_{k-1} = p  &\quad\quad\text{(continuity in space - pauses)} \label{ass:continuity-in-space-pauses}
\end{align}

The second set of properties is meant to provide greater clarity. The first requires that flights last one unit of time which explicitly conditions all the results produced by the model to the selected temporal resolution. Note that fixing the flight duration does not in any way limit modeling velocity since flight lengths (or, more specifically, their distribution) can still be freely chosen. The second avoids confusion with regard to the lengths of pauses by eliminating the possibility of two or more pauses are consecutive, as then they could be combined into a single longer pause. This can be expresses by requiring that the following equalities hold for $k=2, \dots, K$ and almost surely:
\begin{align}
\Delta^T_k &= 1, \quad \text{if } R_k = f &\quad\quad \text{(flights last 1 unit of time)} \label{ass:flights-have-length-1} \\
R_k &= f, \quad \text{if } R_{k-1} = p &\quad\quad  \text{(no consecutive pauses)}\label{ass:no-cons-pauses}
\end{align}

Finally the last two modeling assumptions can be viewed as modeling choices, partially inspired by \citep{barnett2020inferring}.

\begin{align}
   &q(\bM_k|\bM_{k-1}, \dots, \bM_1) = q(\bM_k|\bM_{k-1}) & \text{(Markovianity)} \label{ass:Markov} \\
   &q(\Delta^S_k = \Delta^S_{k-1}| R_{k} = p) = 1 & \text{(pauses store flight information)} \label{ass:DeltaS}
\end{align}

While assumption \eqref{ass:Markov} is relatively straightforward, assumption \eqref{ass:DeltaS} might seem counter-intuitive. It is introduced in order to preserve the relationship between two flights that are separated by a pause; when a person stops, the direction of their subsequent flight may depend on that of their previous flight, which would be impossible if we only imposed an order 1 Markov structure or if we defined $\Delta_k^{S}=0$ whenever $R_k=p$. Thus, when $R_k = p$, we give up the interpretation of $\Delta^S_k$ as a change in space and instead we use it to store the information about the most recent flight.

\begin{proposition}\label{thm:complete-data-likelihood}
The likelihood function corresponding to the FPM indexed by a collection of unknown parameters $\bftheta$ can be written as
\begin{align}
\begin{split}\label{eq:complete-data-likelihood}
    \mathcal{L}(\bftheta \vert \mathcal{M}) &= q(\mathcal{M}| \bftheta) \\
    &= q(\bM_1|\bftheta) \cdot \prod_{R_k = f, k>1}  q(\Delta^S_k|R_k, \bM_{k-1}, \bftheta) \cdot \prod_{R_k = p, k>1} q(\Delta^T_k|R_k = p, \bM_{k-1}, \bftheta)\cdot\\ 
    &\cdot\prod_{k>1} q(R_k| \bM_{k-1}, \bftheta)
\end{split}
\end{align}
\end{proposition}

\begin{figure}[ht]
    \centering
    \includegraphics[width=1.0\textwidth]{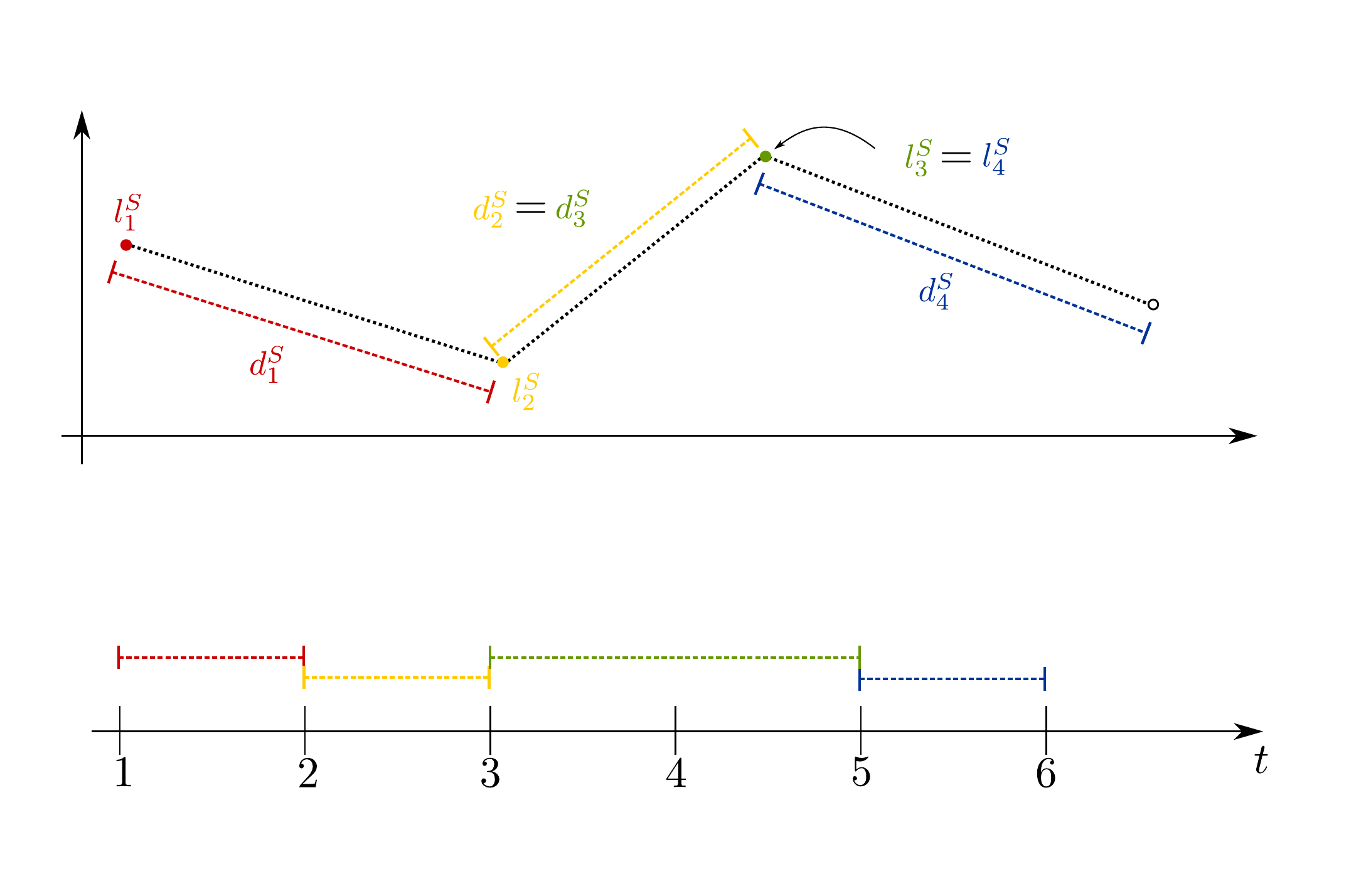}
    \caption{Illustration of the increments from Example \ref{ex:toy}.}
    \label{fig:increments-toy}
\end{figure}

\begin{example}\label{ex:toy}
Figure \ref{fig:increments-toy} shows a realization of a motion $\mathcal{M} = \left(\bM_1, \bM_2, \bM_3, \bM_4\right)$ consisting of three flights and one pause with components denoted as: 
\begin{align*}
\left( \right.&((l^S_1, 1), (d^S_1, 1), f), \\
&((l^S_2, 2), (d^S_2, 1), f), \\
&((l^S_3, 3), (d^S_3, 2), p), \\
&\left.((l^S_4, 5), (d^S_4, 1), f)\right).
\end{align*}
Notice that \eqref{ass:continuity-in-time} results in $l^T_i = l^T_{i-1} + d^T_i$ for $i=2, 3, 4$, while \eqref{ass:continuity-in-space-flights} leads to, for example, $l^S_2 = l^S_1 + d^S_1$. We also have $l^S_4 = l^S_3$ due to \eqref{ass:continuity-in-space-pauses}. Condition \eqref{ass:flights-have-length-1} is reflected in the fact that increments 1, 2 and 4 last one unit of time. Finally, due to \eqref{ass:DeltaS} we have that $d_2^S = d^S_3$.
\qed
\end{example}

Several plots with realizations of motion simulated using our model can be found in the supplementary material. 

Note that \eqref{eq:FPM-basic} leaves much room for further application-dependent specification. For example if at a given location the individual is suspected to make flights in certain directions then this restriction can be incorporated as a particular choice of $q(\Delta^S_k|R_k=f, \bM_{k-1})$, since $\bM_{k-1}$ contains information about the location. Similarly, we can specify $q(\Delta^S_k|R_k=f, \bM_{k-1})$ or $q(\Delta^T_k|R_k=f, \bM_{k-1})$ in a way that reflects the dependence of $\Delta^S_k$ or $\Delta^T_k$ on time (i.e. of the day or of the week, or circadian patterns to movement).

\subsection{Standard parametrization}\label{subsec:standard-param}

As noted above, further selection of the model components might be application-dependent we pursue development under a particular (and intentionally simple) specification of the yet undetermined distributions of \eqref{eq:complete-data-likelihood}. In particular we assume the following:
\begin{enumerate}
    \item \label{spec:four-dim} Parameter $\bftheta = (\theta_1, \theta_2, \theta_3, \theta_4)$ is four-dimensional.
    \item \label{spec:markov-in-types} Increment type depends only on the type of the previous increment (more generally, it could quite conceivably depend on where the individual is located, i.e. on $L$). This is equivalent to saying that $q(R_k | \bM_{k-1}, \bftheta) = q(R_k|R_{k-1}, \bftheta)$
    \item \label{spec:constant-pause-prob} constant probability of pausing after a flight, i.e. $q(R_k=p|R_{k-1} = f, \bftheta) = \theta_1$.
    \item \label{spec:geom} The distribution of pause lengths is geometric, or $q(\Delta^T_k|R_k = p, \bM_{k-1}, \bftheta) = (1-\theta_2)^{\Delta^T_k}\theta_2$.
    \item \label{spec:gaussian-increments} The distribution of $\Delta^S_k$, i.e. the flight's length and direction, depends only on the previous flight's length and direction (as opposed to, for example, the location $L_k^S$) and is normal, independent in each spatial coordinate. Formally,
    $$
        q(\Delta^S_k|R_k, \bM_{k-1}, \bftheta) = \prod_{i=1, 2} \normal((\Delta^S_{k})_i; \mu_i, \sigma),
    $$
    where $(x)_i$ denotes the $i$-th coordinate of vector $x$ and $\normal(\cdot; \mu, \sigma)$ stands for the pdf of a normal distribution with mean $\mu$ and standard deviation $\sigma$. We take $\mu_i = \theta_3(\Delta^S_{k-1})_i$ and $\sigma=\theta_4$.
\end{enumerate}
We say that the assumptions above constitute the \emph{standard parametrization} of the flight-pause model. Sample trajectories generated using this parametrization are shown in Figure \ref{fig:sample-trajectories}.

\begin{proposition}\label{thm:standard-param-likelihood}
Under the standard parametrization the complete-data likelihood for the flight pause-model takes the form
\begin{multline}
    q(\mathcal{\bM}|\bftheta) = q(\bM_1| \bftheta) \cdot \theta_1^{|\mathcal{P}|} (1-\theta_1)^{|\mathcal{F}^f|} \cdot \theta_2^{|\mathcal{P}|}  \left(1-\theta_2\right)^{\sum_{k : R_k=p}\Delta^T_k}\cdot \\ 
    \cdot \prod_{k > 1} \prod_{i=1, 2} \normal((\Delta^S_k)_i; \theta_3(\Delta^S_{k-1})_i, \theta_4),
\end{multline}
where $\mathcal{F}^f = \{k : R_k = f \land R_{k+1} = f\}$ is a set of indices of those increments which are flights and which are followed by other flights and $\mathcal{P} = \{k : R_k = p\}$ is the set of indices of all pauses and $\land$ stands for logical conjunction (``and''). We assume that the first increment is a flight, i.e. $R_1 = f$.
\end{proposition}

The proof of Proposition \ref{thm:standard-param-likelihood} can be found in the Appendix.

\begin{figure}[ht]
    \centering
    \includegraphics[width=1.0\textwidth]{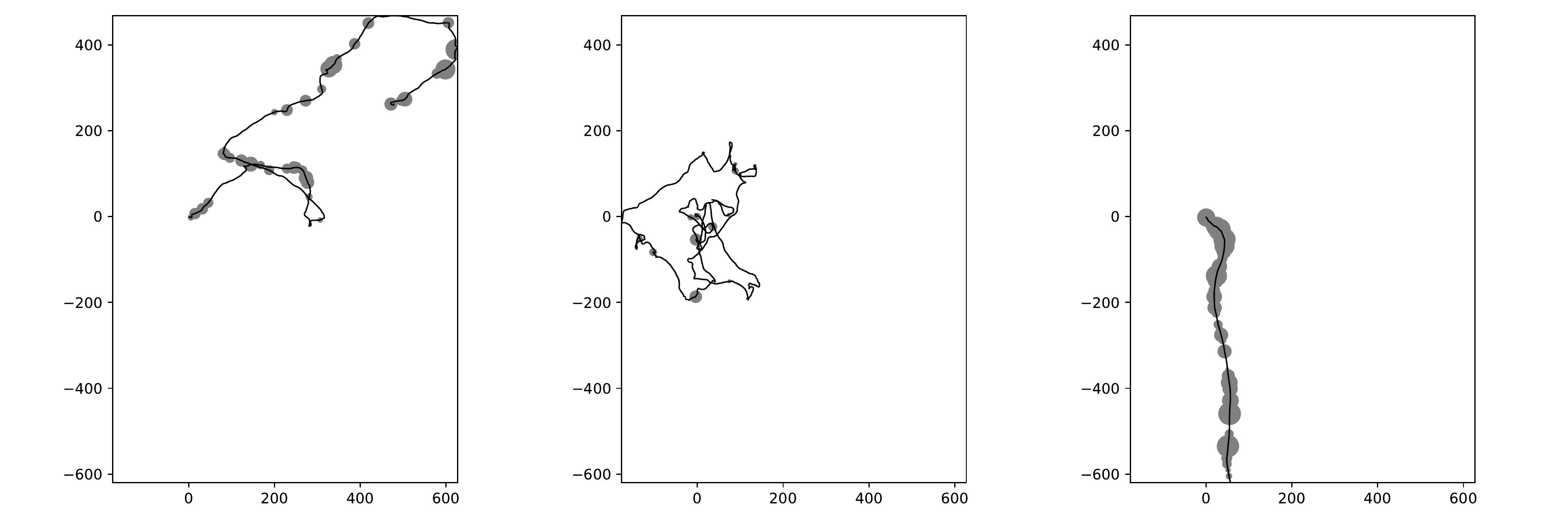}
    \caption{Sample trajectories corresponding to motions generated using the flight-pause model with standard parametrization. In the first panel we used $\bftheta^{(1)} = (0.1, 0.1, 0.95, 1)$, in the second $\bftheta^{(2)} = (0.01, 0.2, 0.9, 1)$ and in the third $\bftheta^{(3)} = (0.8, 0.5, 0.999, 1)$. The grey circles indicate pauses with the center of the circle corresponding to its location and radius proportional to its temporal duration. All trajectories consist of 1000 increments.}
    \label{fig:sample-trajectories}
\end{figure}

\section{Observed data model for incomplete trajectory data}\label{sec:data-model}

In Section \ref{sec:model} we specified the FPM for motions made up of increments. In practice, however, MPT data is not collected directly as increments. Instead, MPT devices are designed to measure the location of the device at certain (typically evenly spaced) points in time. Thus, raw MPT data must be transformed into increments before the FPM can be fitted.  

In situations where locations are not fully observed -- see Section \ref{sec:data-collection} for a list of scenarios that produce gaps in observed MPT trajectories -- it may not be possible to transform observed locations into increments. This creates a somewhat unusual situation in which the statistical model is specified in such a way that some portion of the observed data do not contain any relevant information about the unknown model parameters because other missing observations preclude calculation of the transformed data.\footnote{We note a somewhat analogous situation arises when in the analysis of time-series data using Auto-Regressive Integrated Moving Average Models \citep[ARIMA, e.g.][]{brockwell2009time} in the presence of missing data. In this case, lagged differences in the variable cannot always be calculated fully from the observed data.}  

We start this section by describing the tools that are needed to connect MPT data (i.e., time-stamped locations in geographic space) to increments defined in Section \ref{sec:model}. We build on the notation introduced in Section \ref{sec:model}, following standard missing-data notational conventions. We note that extending the FPM framework to accommodate missing data is not immediate, however, because of a complex relationship between the locations and the increments. In particular, MPT dictates a data collection mechanism for locations, but inference in the FPM requires the formulation of a corresponding data collection model for increments. We account for these complexities and establish the connection between observability of these two random objects.

\subsection{Expressing motion as a sequence of locations}

To make concrete the connection between motion and location sequences, consider a single increment $\bM = ((L^T, L^S), (\Delta^T, \Delta^S), R)$, dropping momentarily the $k$ subscript for clarity.  We define $\mathcal{D}(\bM)$ to be the set of time points from the beginning of $\bM$ up until its end. Formally, 
\begin{equation*}
    \mathcal{D}(\bM) \equiv \begin{cases}
      (L^T), & \text{ if } \;R = f, \\
      \left(L^T, L^T + 1, \dots, L^T + \Delta^T - 1\right), & \text{ if } \;R = p.
    \end{cases}
\end{equation*}
Notice that $\lvert \mathcal{D}(\bM) \rvert = \Delta^T$, i.e. that the number of elements in $\mathcal{D}(\bM)$ is equal to the duration of that increment.

We can define an analogous concept for the sequence of spatial coordinates underlying an increment, which we call the \emph{trajectory} of $\bM$ and write $\tau(\bM)$. Since flights last only one unit of time, their trajectory should correspond only to the original location $L$. For pauses, during which physical location does not change but which might last several units of time,  the trajectory contains multiple copies of the same location. We can thus formally define a trajectory as 
\begin{equation*}\label{eq:increment-trajectory}
    \tau(\bM) \equiv \begin{cases}
      \left(L^S\right), & \text{ if } \;R = f, \\
      \underbrace{\left(L^S, \dots, L^S\right)}_{\Delta^T \text{ times}}, & \text{ if } \;R = p.
    \end{cases}
\end{equation*}

Since a motion is composed of increments, we now naturally extend the concept of a trajectory from increments to a motion. Specifically, slightly abusing notation, we define $\tau(\mathcal{M}) \equiv \cup_{k = 1}^{K} \tau(\bM_k)$, to be the concatenation of the trajectories of all increments in $\mathcal{M}$. Thus, the trajectory of a motion is essentially a sequence of random spatial coordinates indexed by time. We can also write $\tau(\mathcal{M}) = \left(\mathcal{S}_t\right)_{t \in \mathbb{N}}$, where each $\mathcal{S}_t$ are collections of random spatial coordinates, each of which can be interpreted as the position of the MPT device at time $t$.

These new concepts can be illustrated using Example \ref{ex:toy} in Figure \ref{fig:increments-toy}. In this example, the realization of $\left(\mathcal{S}_t\right) = \tau(\mathcal{M})$ is $\left(l_1^S, l_2^S, l_3^S, l_4^S\right)$. Similarly, $\mathcal{D}(\bM_1)$ is realized as ${1}$, while $\mathcal{D}(\bM_3)$ as $(3, 4)$.

To be precise about the connection between the motion and a trajectory we formulate
\begin{proposition}\label{prop:unique-trajectory}
For each realization of the motion $\mathcal{M}$ there exists a unique realization of trajectory $\left(\mathcal{S}_t\right)_{t \in \mathbb{N}}$.
\end{proposition}
\noindent
The proof can be found in the Appendix.

\subsection{From observed locations to observed increments}\label{subsec:observed-locations}

As a first step towards defining a data collection mechanism for increments, we introduce an \emph{observability indicator} for locations comprising a trajectory. The \emph{observability indicator} for $\mathcal{S}_t$ is defined as
$$
Z_t = 
  \begin{cases}
    1, & \text{if } \mathcal{S}_t \text{ is observed}, \\
    0, & \text{otherwise}.
  \end{cases}
$$
We let $\bZ = \left[Z_1, Z_2, \dots\right]^\top$ denote the vector of observability indicators indexed by time.  This notation follows a standard convention used in the missing data literature and assumes that we know precisely which locations are observed and which ones are not.  That is, inference on the $\bftheta$s in the FPM should be conditioned on the observed increments and the realized values of $\mathcal{S}$ and $\bZ$. 

We can use a similar notation to distinguish between the observed and missing increments. Specifically, we define $\xi_k = 1$ if increment $\bM_k$ is observed and $\xi_k=0$ otherwise and we say that $\xi_k$ is the \emph{observability indicator} for $\bM_k$. Similar to the notation introduced for locations we use $\bfxi = [\xi_1, \xi_2, \dots]$ to denote the sequence of all observability indicators ordered to correspond to the increments that comprise $\mathcal{M}$. This simple notation, however, belies a more complex reality of what it takes for an increment to be observed. In addition to an increment's origin in space and time, we must also know its duration and spatial displacement for the increment to be fully observed.  Therefore, several consecutive locations need to be observed. This requirement is made precise in the following:
\begin{proposition}\label{prop:obs-increments}
Let $\bM_k = (L_k, \Delta_k, R_k)$ be an increment contained in  motion $\mathcal{M}$. We have 
$$
\xi_k = 1 \iff \begin{cases} R_k = f &\land \quad \prod_{t \in \mathcal{D}(\bM_k) \cup \floor{\mathcal{D}(\bM_{k+1})}} Z_t = 1,\\
  R_k = p &\land \quad \prod_{t \in \mathcal{D}(\bM_k) \cup \mathcal{D}(\bM_{k+1}) \ceil{\mathcal{D}(\bM_{k-1})} \cup \floor{\mathcal{D}(\bM_{k+2})}} Z_t= 1,
\end{cases}
$$
where  $\floor{A}$, $\ceil{A}$ stand for the first and last element of an ordered set $A$, respectively.
\end{proposition}
In other words, we observe a flight when we know the two locations at its beginning and end (i.e. its trajectory and the beginning of the trajectory of the next increment). In order to observe a pause, we have to have access to all of its trajectory as well as the locations immediately proceeding and following it (i.e. the first and last elements of the trajectories of the preceding and succeeding increments, respectively).

We note here a key distinction between the observability indicators $\bZ$ and $\bfxi$.  For locations, $Z_t$ is defined and observed at every time increment; it is always known whether a location at a given $t$ was recorded.  In contrast, the observability indicator for increments, $\bfxi$, is not defined for every $t$. In fact, it requires different indexing for it will not generally be known how many increments are unobserved, as a sequence of missing locations may include a pause $\bM_k$ of length $\Delta^T_k>1$. The consequences of this fact and an alternative indexing scheme for observed increments are discussed in Section \ref{subsec:observed-data-likelihood}.

Finally we note that Proposition \ref{prop:obs-increments} implicitly assumes that the only reason why increments might be missing is due to certain locations along the trajectory begin unobserved. One could conceive of other reasons for missingness not directly linked to the function of MPT measurements (i.e. software error which randomly removed some previously recorded increments). Unless stated otherwise, we henceforth assume that missing increments are solely a consequence of missing locations.

\subsection{Effective sample size}\label{subsec:eff-samp-size}

The distinction between a motion (the object used to construct the FPM) and its trajectory (the set of points recorded by MPT) invites a corresponding distinction between two concepts of a sample size. Since most MPT devices record realized trajectories, it might seem intuitive to use the number of recorded locations as the sample size. In fact this is the number that is frequently reported \citep[e.g.][]{rhee2011levy}. At the same time, what is necessary for evaluating the FPM likelihood are increments. This motivates the following definition:
\begin{definition}\label{def:eff-samp-size}
We say that effective sample size is equal to $\sum_{k=1}^K \xi_k$, i.e. the number of observable increments.
\end{definition}
It follows that when a large number of locations are observed, the effective sample size-- i.e., information that can be used to infer $\bftheta$ might be small. For example, an MPT measurement scheme that collects every other location (i.e. $\forall l \in \mathbb{N}$, $Z_{2l} = 0$ and $Z_{2l+1} = 1$) will never record an increment, resulting in effective sample size $\sum_{k=1}^K \xi_k = 0$. As a consequence, the likelihood function \eqref{eq:complete-data-likelihood} cannot be calculated and the data have no information with which to make inference on the parameters of a flight-pause model.

\subsection{Observed-data likelihood}\label{subsec:observed-data-likelihood}

We now introduce notation for observable increments.  To begin, we note that every possible value of $\bfxi$ can be composed of alternating groups of consecutive observed and unobserved increments. 
Formally, we let $\mathcal{O} = \{\bM_k \in \mathcal{M} : \xi_{\bM} = 1\}$ and define $O= \{O_1, \dots, O_J\}$, a partition of $\mathcal{O}$ ordered from the smallest (1) to largest time index ($L_k^T$). We say that the index of the first increment in $O_j$ is $I_j$. In other words $\floor{O_j} = \bM_{I_j}$. We also define $N_j = |O_j|$, the number of increments in the $j$-th observed block. This means that the last location in the observed block of increments is $\ceil{O_j} = \bM_{I_j + N_j - 1}$. 

We denote the blocks that make up the unobserved increments of the motion as $U_j$, with $\mathcal{U}= \{\bM_k \in \mathcal{M} : \xi_k = 0\}$. Specifically, $U_j = \left(\bM_{I_j + N_j}, \dots, \bM_{I_{j+1} - 1}\right)$. In this way, we have $\mathcal{M} = (O_1, U_1, O_2, U_2, \dots, O_{J-1}, U_{J-1}, O_J)$. Using this notation, as well as $\bfpsi$ to denote the additional parameters related to the distribution of $\bfxi$ ,we can now write the observed data likelihood in terms of blocks of observed increments:
\begin{multline}\label{eq:obs-data-integral}
q(\mathcal{O}, \bfxi|\bftheta, \bfpsi) = \int q(\mathcal{O}, \mathcal{U}|\bftheta) q(\bfxi|\mathcal{O}, \mathcal{U}, \bfpsi) d\mathcal{U} = q(O_1| \bftheta) \prod_{j>1}^J \prod_{k = I_j + 1}^{I_j + N_j - 1}  q(\bM_k|\bM_{k-1}, \bftheta) \\ \int \prod_{j>1}^J \prod_{k=I_j + N_j}^{I_{j+1} - 1} q(\bM_k|\bM_{k-1}, \bftheta) q(\bM_{I_j}|\bM_{I_j-1}, \bftheta) q(\bfxi|\mathcal{O}, \mathcal{U}, \bfpsi) d\mathcal{U}.
\end{multline}
Expression \eqref{eq:obs-data-integral} constitutes the general form of the observed data likelihood, where the first set of products contains indices of increments contained within $\mathcal{O}$ and the second set of products has indices corresponding to increments contained within $\mathcal{U}$.  Appendix \ref{app:inference-obs-data} includes derivation of the special case of \eqref{eq:obs-data-integral} under the standard parametrization introduced in Section \ref{subsec:standard-param}.

Note the explicit dependence of \eqref{eq:obs-data-integral} on $q(\bfxi|\mathcal{O}, \mathcal{U}, \bfpsi)$, which corresponds to the data collection mechanism for increments - a probability model dictating which elements of a motion are observed.  Dependence of (\ref{eq:obs-data-integral}) on this quantity clarifies that, in general, inference in the FPM with incomplete trajectory data will require assumptions about the form of this data collection mechanism.

Before discussing data collection mechanisms for increments, we introduce an assumption that can drastically simplify evaluation of  \eqref{eq:obs-data-integral}: 
\begin{assumption}\label{ass:indep-blocks}
Assume that for each $j = 1, \dots, J$ and $\bM_{k_1}, \bM_{k_2}$ such that $\bM_{k_1} \in O_j$ and $\bM_{k_2} \in \mathcal{M} \setminus O_j$ we have $\xi_{k_1} \perp \xi_{k_2} \;|\; O_j \cup \left\{\xi_k : \bM_k \in O_j \right\}$. Assume further that if $\bM_{k_1}, \bM_{k_2}$ are such that $\bM_{k_1} \in O_j \cup U_j$ and $\bM_{k_2} \in \mathcal{M} \setminus \left(O_j, \cup U_j\right)$ then $\bM_{k_1} \perp \bM_{k_2}$.
\end{assumption}
This assumption implies that the observed blocks are independent and that the observability indicators depend only on other increments within the same block \citep[see e.g. for a similar approach][]{deChaumaray2020mixture}. 

Thus under \eqref{ass:indep-blocks}, we have $q(\bfxi|\mathcal{O}, \mathcal{U}) = \prod_{j=1}^{J-1} q(\bfxi_{O_j, U_j}|O_j, U_j) q(\bfxi_{O_J}|O_J)$ and \eqref{eq:obs-data-integral} can be expressed as
$$
q(\mathcal{O}, \bfxi | \bfpsi, \bftheta) = \prod_{j=1}^J q(\bM_{I_j}|\bftheta)q(\bfxi_{I_J}=1|I_j, \bfpsi) \prod_{k = I_j}^{I_j + N_j - 1} q(\bM_k|\bM_{k-1}),
$$
where $\bfxi_{O_j, U_j}$ and $\bfxi_{O_j}$ stand for, respectively, the vector of observability indicators for the elements of the blocks $O_j$ and $U_j$ jointly and just $O_j$. Intuitively, this means that we treat each observed block as a distinct trajectory depending on a common set of parameters. Therefore, \eqref{ass:indep-blocks} can be a good approximation of the truth if the unobserved blocks are large (i.e., if $I_{j+1}-I_j -N_j$ is large). Evaluating the likelihood under \eqref{ass:indep-blocks} is equivalent to calculating the composite likelihood with identical weights \citep{lindsay1988composite, varin2011overview}. 

\section{Data collection mechanisms}\label{sec:data-collection}

The previous section introduces a framework for studying the implications of data collection mechanisms $q(\bfxi | \mathcal{M}, \bftheta, \bfpsi)$. We now explore different missing-data mechanisms (i.e., models for $\bfxi$) and their implications for inference on $\bftheta$. 

To study the impacts of possible data collection mechanisms, we first briefly review the classic missing data framework, as described \citet{little2019statistical} and \citet{gelman2013bayesian}. Consider the joint likelihood,
\begin{equation}\label{eq:likelihood-increments-indicators}
    q(\mathcal{M}, \bfxi| \bftheta, \bfpsi) = q(\mathcal{M}| \bftheta)q(\bfxi|\mathcal{M}, \bfpsi, \bftheta),
\end{equation}
where the first term on the right-hand side is given in Proposition \ref{thm:standard-param-likelihood} and the equality holds because $\bfpsi$ are used to parametrize only the distribution of $\bfxi$. We use the second term, which we call data-collection mechanism, to express various assumptions regarding the observation pattern. 

We start with a simple example of a mechanism, which might sometimes be used to conserve battery but which actually results in no missing data.

\begin{example}[(Movement-triggered data collection)]
Consider a mechanism that starts collecting data once movement is detected (i.e. using accelerometer which is found in most modern smart phones) and stops when a pause in movement is detected. In particular, we assume that 
$$
q(\bfxi|\mathcal{M}, \bfpsi, \bftheta) = \prod_{k=1}^K q(\xi_k|\bM_k),
$$
where $q(\xi_k|\bM_k) = \mathbbm{1}(R_k=f)$. Notice that even though we technically suspend data collection for the duration of pauses, knowing that the pauses are the only unobserved increments actually allows us to have complete information about them, i.e. no data are missing.
\end{example}

\subsection{Ignorable mechanisms}

The key distinction in the study of missing data is between contexts in which the data collection mechanism is \textit{ignorable} and the ones in which it is not. Ignorability is equivalent to two conditions. The first one, called parameter distinctness, requires that the second term on the right hand side of \eqref{eq:likelihood-increments-indicators} takes the form
$$
q(\bfxi|\mathcal{M}, \bfpsi, \bftheta) = q(\bfxi|\mathcal{M}, \bfpsi),
$$
which means that the parameters governing the model for the complete data are distinct from the parameters regulating the data collection mechanism.

The second condition, which is often expressed by saying that the data is \textit{missing at random} (MAR), demands that the data collection mechanism is also independent of the unobserved variables. Mathematically, within the context of our model, this means that
$$
q(\bfxi|\mathcal{O}, \mathcal{U}, \bfpsi, \bftheta) = q(\bfxi|\mathcal{O}, \bfpsi, \bftheta),
$$
since $\mathcal{M} = \mathcal{O} \cup \mathcal{U}$.
An important special case of this condition, called \textit{missingness completely at random} (MCAR), further constrains the data collection mechanism to be independent of the data, i.e. $q(\bfxi|\mathcal{M}, \bfpsi, \bftheta) = q(\bfxi|\bfpsi, \bftheta)$. 

In summary, in the FPM, the data collection mechanism is ignorable if and only if
$$
    q(\bfxi|\mathcal{M}, \bfpsi, \bftheta) = q(\bfxi|\mathcal{O}, \bfpsi).
$$
To provide concreteness, the following example describes a somewhat unrealistic ignorable data collection mechanism for MPT data. 
\begin{example}[(Random increment corruption)]
Consider a data collection scheme where timestamped locations are measured then, after processing and storing these locations as increments,  buggy software leads to the deletion of some of the increments in $\mathcal{M}$. For example after every $l$ recorded increments, there is one increment missing. Then another $l$ increments are recorded etc. In this case $q(\bfxi|\mathcal{M}, \bftheta, \bfpsi) = q(\bfxi|\bfpsi)$ so data mechanism is ignorable. Note that in this case increments are missing for reasons unrelated to the availability of the underlying locations.\qed
\end{example}

If a mechanism is ignorable, inference of $\bftheta$ can be conducted without explicitly modeling the data collection mechanism.  That is, inference on $\bftheta$ can be obtained using the complete data likelihood:
\begin{align*}
q(\mathcal{O}, \bfxi|\bftheta, \bfpsi) &= \int q(\mathcal{M}, \bfxi| \bftheta, \bfpsi) d\mathcal{U} = \int q(\mathcal{M}| \bftheta) q(\bfxi | \mathcal{M}, \bfpsi) d\mathcal{U} = \\
&= \int q(\mathcal{O}, \mathcal{U} | \bftheta) q(\bfxi|f(\mathcal{O}), \bfpsi) d\mathcal{U} \\
&\propto \int q(\mathcal{O}, \mathcal{U} | \bftheta) d\mathcal{U}.
\end{align*}

\subsection{Non-ignorable mechanisms}\label{subsec:non-ignorable}

In the context of MPT data, many common data collection mechanisms turn out not to be ignorable for the FPM. Consider the following simplistic example.

\begin{example} 
An MPT instrument records flights with probability $\rho_f$ and pauses with probability $\rho_p \neq \rho_f$.  It follows that, 
$$q(\bfxi|\mathcal{\bM}) = \prod_{k=1}^K q(\xi_k|R_k) = \prod_{k=1}^K \rho_p\mathbbm{1}(R_k=p) + \rho_f\mathbbm{1}(R_k=f).$$ 
Since the probability of observing an increment depends on its potentially-unobserved increment type, the data are not missing at random and the mechanism is not ignorable for the FPM. \qed
\end{example}

More realistically, an MPT device might collect \emph{locations} only during certain prescribed intervals in time (e.g., alternating between recording for one minute and not recording for one minute) to save battery power. \citet{barnett2020inferring} observe that this approach is used by a popular Beiwe app \citep{torous2016new} and study its different versions.

\begin{example}[(The on-off mechanism)]\label{ex:onoff}
An MPT device alternately records locations for some prescribed time interval $o$ and then the collection is suspended for another time interval $u$. We show that such a scheme leads to increments which are MNAR for the FPM. 

In order to describe this mechanism formally, for a time $t \in \mathbb{N}$ define $B_t = \floor*{\frac{t}{o + u}}$ and $B'_t = B_t + o$. Notice that when $\floor{\cdot}$ is applied to a totally ordered set or sequence it denotes their minimal element, while when applied to a real number it denotes the largest integer smaller than that number. 
Then let $z_t = 1$ for $t \in \left[B_t, B'_t\right)$ and $z_t = 0$ when $t \in \left[B_t', B_t'+u\right)$. An illustration of this pattern of $o = u$ is shown in Figure \ref{fig:simulations-obs-pattern}.

With this observation scheme, only pauses shorter than $o$ can be observed. This is because we never observe more than $o$ consecutive locations. Therefore, even if we happen to observe the start a long pause $\bM_k$ and its immediately succeeding flight $\bM_{k+1}$, we would not be able to observe all of its trajectory $\tau(\bM_k)$, because $|\mathcal{D}(\bM_k)|>o$ By Proposition \ref{prop:obs-increments} this means that $\bM_k$ could not be observed.\qed
\end{example}

It is worth pointing out that in Example \ref{ex:onoff} as in several others throughout this section and elsewhere in the literature the data collection mechanism is defined for locations. In other words, assumptions about missingness pertain to the random observation indicator $\bz$ and not $q(\bfxi|\bfpsi, \mathcal{M})$. Obviously, these two mechanisms are closely related. However, since the FPM is a model for motions and their increments, assumptions such as ignorability must be formualted in terms of $q(\bfxi|\bfpsi, \mathcal{M})$. Thus, even though often times the missing data mechanism for MPT data may be assumed ignorable with respect to locations, the resulting missing data mechanism for increments may not be ignorable. To see this more clearly consider the following unrealistic but illustrative 

\begin{example}\label{ex:locationsMAR-unrealistic} Assume we pause data collection once $\mathcal{S}_t = \mathcal{S}_{t+1}$ (i.e. once two consecutive spatial locations are identical) and resume it at $t + c$ for $1 < c \in \mathbb{N}_{+}$. We see that in such a scheme locations are MAR but increments are MNAR because no pauses will ever be observed. The reason for this discrepancy is that the mechanism for collecting increments is turned off once it determined it is in the middle of collecting a pause. \qed
\end{example}

We conclude with a more realistic version of Example \ref{ex:locationsMAR-unrealistic}.

\begin{example}[(Geometric gaps)]\label{ex:geometric-gaps}
Consider a malfunctioning MPT data collection device that records every location with probability $\eta = 1-\epsilon$, where $\epsilon$ is some small positive number. Under this data collection mechanism,  the length of the observed and unobserved blocks are random, each following a geometric distribution with success probabilities $\epsilon$ and $\eta$, respectively.
This data collection scheme is MNAR for the FPM. To show this, 
we can calculate the probability that a pause $\bM_k$ is observed as
$$
q(\xi_k=1|\bM_k) = q(\xi_k=1|\Delta^T_k) = \eta^{\Delta^{T_k+3}}.
$$
Similarly, the probability of observing a flight is
$$
q(\xi_k = 1|\bM_k) = \eta^2.
$$\qed
\end{example}

In general, if the mechanism is not ignorable, inference requires that we calculate the entire integral $\int q(\mathcal{M}| \bftheta) q(\bfxi | \mathcal{M}, \bfpsi) d\mathcal{U}$. Proposition \ref{thm:observed-data-likelihood} shows how this can be done under the standard parametrization in the case of the data collection mechanism described in Example \ref{ex:onoff} and Example \ref{ex:geometric-gaps}. 

\begin{proposition}\label{thm:observed-data-likelihood}
Under the standard parametrization, using the symbols defined in Section \ref{app:inference-mnar}, \eqref{ass:indep-nablas} and assuming the data collection mechanism as in Example \ref{ex:onoff} the log of the observed data likelihood can be expressed as
\begin{multline}\label{eq:log-obs-data-standard}
    \log q(\mathcal{O}|\bftheta) = \left(\sum_{j=1}^J|P_j| + \sum_{j=1}^{J-1} \delta_j \right) \log\theta_1 + 
        \left( \sum_{j=1}^J |F_j| \right) \log\left(1-\theta_1\right)
        + \left(\sum_{j=1}^J|P_j| + \sum_{j=1}^{J-1}\gamma_j\right)\log\theta_2 \\
        + \left( \sum_{k \in P_j} \Delta^T_k  - \sum_{j=1}^J|P_j| + \sum_{j=1}^{J-1}d_j + g_{j+1}\right)\log\left(1-\theta_2\right) + \sum_{j=1}^{J-1} \log Q_{\delta_j+1, \gamma_{j+1}+ 1}(|U_j|) + \\
        + \sum_{j=1}^J \log q\left(\Delta^S_{I_j}|\bftheta\right) -\frac{|\mathcal{F}\cap O_j| - 1}{2}\log \left(2\pi\theta_4\right) - \frac{1}{2\theta_4} \sum_{k = I_j + 1}^{I_j + N_j - 1} \left(\Delta^S_k - \theta_3\Delta^S_{k-1}\right)^2
\end{multline}
where $Q_{m,l}(n)$ is the $m,l$-th entry of the matrix $\frac{1}{\theta_1 + \theta_2}\left(\left[\begin{array}{cc}\theta_2 & \theta_1 \\ \theta_2 & \theta_1\end{array}\right] + (1 - \theta_1 - \theta_2)^n \left[\begin{array}{cc}\theta_1 & -\theta_1 \\ -\theta_2 & \theta_2\end{array}\right]\right)$.
\end{proposition}

The proof of this proposition can be found in the Appendix and relies on representing the increments' durations and types using a Markov chain and assuming that the direction and length of flights in a given observed block are independent of these properties of flights in other observed blocks.

\section{Motion/trajectory imputation}\label{sec:imputation}

With a formal framework to account for the mechanism dictating which increments are recorded, the observed data likelihood and the data collection model give us a practical tool to estimate the model parameters, which are generally unknown. This, in turn, opens the door to generating (imputing) the missing parts of the trajectory under the FPM. 

\subsection{Motivation}\label{subsec:imputation-motivation}

There are several reasons why imputation might be of interest. For example, a researcher may be using partially observed MPT data to measure exposure to some phenomena \citep{yi2019methodologies} associated with a specific point or area in geographic space\footnote{Estimating the duration of an individual's exposure to geographically-referenced exposure source is a common problem in environmental health \citep[e.g.][]{lippmann2009environmental, henneman2021comparisons}, social determinants of health \citep[e.g.][]{braveman2011social, viner2012adolescence}, and contextual effects \citep[e.g.][]{alexander1975contextual, browning2021human, erbring1979individuals} applications.}.  If total exposure ``dose" is assumed to be proportional to time spent at a point or in an area \citep[e.g.][]{nyhan2019quantifying}, the total dose may be underestimated if the missing MPT data are ignored. We focus on this perspective throughout the remainder of the paper. Alternatively, one might be interested in filling in gaps in MPT data for the purpose of visualization, or even simulating a trajectory in an agent-based models \citep[e.g.][]{qiao2018role}. 

\subsection{Imputation algorithm}\label{subsec:imputation-algo}

We now present a plug-in method for generating a single imputation of the missing MPT data under the FPM. At the same time, multiple imputations of the same missing part of the motion might often be desireable \citep{meseck2016missing}, like in the experiment we conduct in Section \ref{sec:simulations}.

Our procedure consists of two steps. First, we use the observed increments $\mathcal{O}$ to calculate the observed data likelihood \eqref{eq:obs-data-integral} and find the vector $\hat{\bftheta}$, for which it attains a maximum.

Second, we generate a draw from $q(\mathcal{U}|\hat{\bftheta}, \mathcal{O})$. The former is fairly straightforward, and a number of standard  optimization approaches can be used. The latter is more complicated. To see why, recall the notation introduced in Section \ref{sec:data-model} which allows us to write
$$
q(\mathcal{U}| \hat{\bftheta}, \mathcal{O}) = \prod_{j=1}^{J-1} q(U_j\;|\;\hat{\bftheta}, \mathcal{O}) = \prod_{j=1}^{J-1} q(U_j\;|\;\hat{\bftheta}, \bM_{I_j + N_j - 1}, \bM_{I_{j+1}}),
$$
where the last equality is due to \eqref{ass:Markov} (Markovianity).
Under standard parametrization each term in the product can be further expressed as

\begin{equation*}
    q(U_j\;|\;\hat{\bftheta}, \bM_{I_j + N_j - 1}, \bM_{I_{j+1}}) = \prod_{k = I_j + N_j}^{I_{j+1} - 1} q(\bM_k\;|\;\hat{\bftheta}, \bM_{k-1}, \bM_{I_{j+1}})
\end{equation*}

It is difficult to sample from this distribution because, as noted at the end of Section \ref{sec:data-model}, the number of increments that need to be sampled is unknown, or equivalently, $\vert U_j \vert$ is unknown for all $j$. Moreover, the exact form of $q(\bM_k|\hat{\theta}, \bM_{k-1}, \bM_{I_{j+1}})$ is also difficult to derive. To overcome these challenges we use the following simplification: we sample increment types from $q(R_k|R_{k-1}, \theta_1)$ and, whenever $R_k = p$, we also sample the pause duration from $q(\Delta^T_k|R_k = p, \hat{\theta}_2).$ We label the increments sampled in this way as $\hat{U_j} = \left(\hat{\bM}_{I_j + N_j}, \hat{\bM}_{I_j + N_j + 1}, \dots \hat{\bM}_{I_j + N_j + d}\right)$. Now if we use $d$ to stand for the total duration of the sampled flights (i.e., $d = \sum_{k=I_j + N_j}^{I_j + N_j + d}\Delta^T_k$), then we can express the criterion for when to stop sampling as 
\begin{equation}\label{eq:stopping} 
L^T_{I_j + N_j - 1} + d \geq L^T_{I_{j+1}}.
\end{equation}
In words, we stop sampling when the time at which the last sampled increment starts immediately precedes the start of first increment in $I_{j+1}$. Note that \eqref{eq:stopping} will hold with equality if $\hat{\bM}_{I_j + N_j + d}$ (the last sampled increment) is a flight (i.e., $R_{I_j + N_j + d}= f$). If $R_{I_j + N_j + d} = p$, then we might need to adjust $\Delta^T_{I_j + N_j + d}$ such that $L^T_{I_j + N_j + d} + \Delta^T_{I_j + N_j + d} = L^T_{I_{j+1}}$. This strategy is followed, for example, by \citet{barnett2020inferring}. One can also use a bridging technique described therein in order to ensure the continuity of the trajectory.

Once the increment types are sampled -- implying the number of flights is known -- we can use the forward filter-backward sampler algorithm (see \citet{Fruhwirth1994, Carter1994, Durbin2002} with a scalable version in \citet{jurek2022ffbs}) to ensure trajectory continuity. This requires representing flights $\mathcal{F}$ as a state-space model, as described in Appendix \ref{app:state-space-model}. 

\section{Numerical simulations}\label{sec:simulations}

In this section, we use simulated data to study certain properties of the FPM and its standard parametrization. To illustrate the importance of properly accounting for the data collection mechanism, we show that inaccurate assumptions about the type of missingness leads to biased inference. Second, we show that the imputation method introduces in Section \ref{sec:imputation}  outperforms the existing methods according to some metrics.

\subsection{Parameter estimation}\label{subsec:simulations-parameters}

We generate motions $\{\mathcal{m}_\eta\}_{\eta=1}^{N_{\text{tr}}}$ with $N_{\text{tr}} = 100$ from the model described in Section \ref{sec:model} under the standard parametrization. We set the time limit $t_{\max} = 1000$ and $\theta = (0.1, 0.1, 0.95, 1)$. We assume that $L^S_1 = (0, 0)$ and that $q(\bM_1) = \mathbbm{1}(L_1^S = (0, 0), L_1^T = 1, \Delta_1^T = 1, R_1 = f)\cdot \normal(\Delta^S_1; 0, 1)$. We then mimic the on-off mechanism (Example \ref{ex:onoff} in Section \ref{subsec:non-ignorable} by masking the locations at certain prescribed intervals. Specifically, we assume that $z_t = 1$ if $t \in \cup_{l \in \mathbb{N}} [2l\cdot o, 2l\cdot o + 1)$, where $o=25$ and that $z_t = 0$ otherwise. Recall that this on-off data collection mechanism, illustrated in Figure \ref{fig:simulations-obs-pattern}, is MNAR.

Using the observed locations and Proposition \ref{thm:observed-data-likelihood}, we calculate the observed increments and find $\hat{\bftheta}$, the value of $\bftheta$ that maximizes the observed data likelihood. A histogram of these estimates for all simulated motions is shown in Figure \ref{fig:param-estimation-u}. For comparison, we also show the distribution of the maximum likelihood estimates calculated under the (incorrect) assumption that the increments are MAR (Figure \ref{fig:param-estimation-b}). These results clearly indicate the need for accounting for the data collection mechanism, as some of the estimates obtained under the MAR assumption exhibit a significant bias.

\begin{figure}
    \centering
    \includegraphics[width=0.6\textwidth]{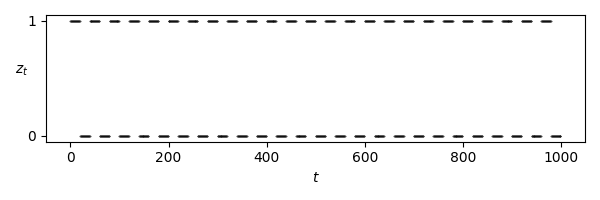}
    \caption{The pattern of missing locations described in Example \ref{ex:onoff} with $o=u=20$. Sometimes, the blocks of unobserved locations, i.e. intervals where $z_t=0$ can be longer than the observed blocks (where $z_t=1$) in order to conserve battery power.}
    \label{fig:simulations-obs-pattern}
\end{figure}

\begin{figure}[ht]
  \centering
  \begin{subfigure}{1.0\textwidth}
    \includegraphics[width=1.0\textwidth]{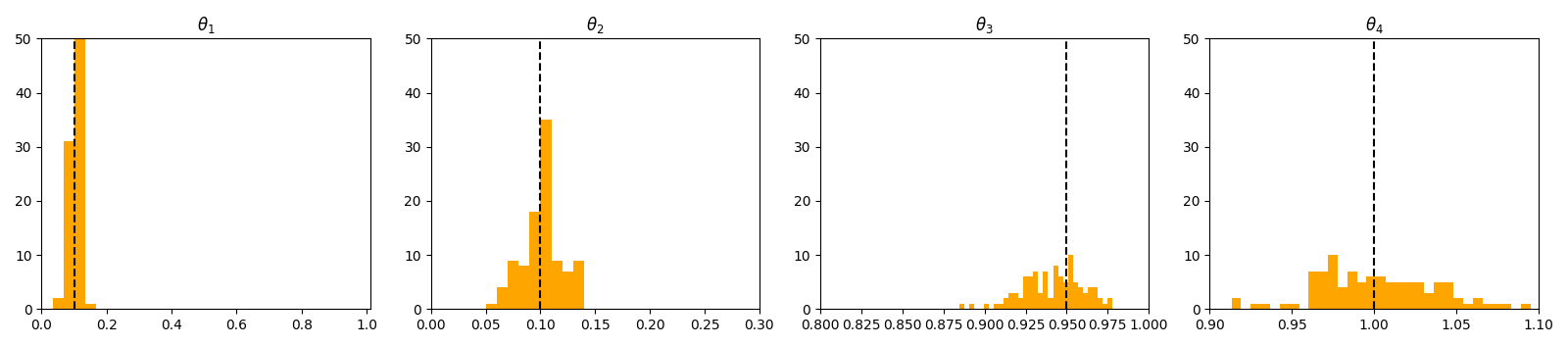}
    \caption{Unbiased estimation}
    \label{fig:param-estimation-u}
  \end{subfigure}
  \begin{subfigure}{1.0\textwidth}
    \includegraphics[width=1.0\textwidth]{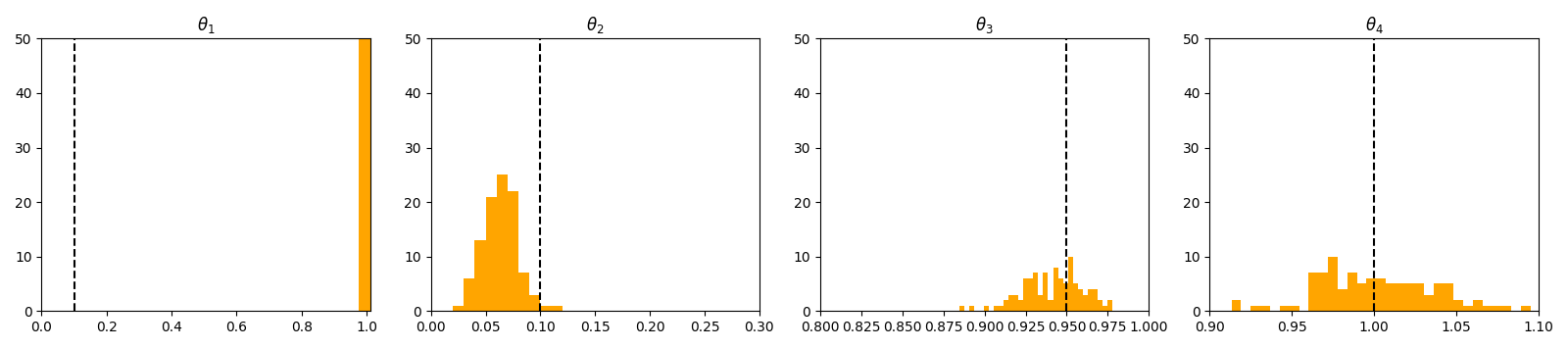}
    \caption{Biased estimation}
    \label{fig:param-estimation-b}
  \end{subfigure}
  \caption{The upper panel shows histograms of the estimates, obtained by maximizing the observed data likelihood, of the components of $\bftheta$ across the simulated motions. The dashed vertical lines indicate the true value. The lower panel displays the histograms of the parameter values $\hat{\bftheta}$ that maximize the observed data likelihood for each of the simulated motions without accounting for the data collection mechanism. The dashed vertical lines indicate the true value. Only the estimates of the first two parameters are biased, which is why the last two histograms (displaying parameters related to flights) look the same, analogous to the pattern in Figure \ref{fig:application-params}. Note that $\theta_1$ is always estimated to be 1 because very few pauses are observed since $|\mathcal{P} \cap \mathcal{O}|$ is very small.}
  \label{fig:param-estimation}
\end{figure}

\subsection{Trajectory interpolation}\label{subsec:simulations-imputation}

In certain situations, for example when evaluating exposure, imputing the missing part of the motion may be of greater importance than parameter estimation. To examine the ability of our model to accopmlish this task we compare the following three imputation methods:
\begin{description}
\item[\textbf{linear interpolation}:] Linear interpolation can be viewed as a special case of the procedure described in Section \ref{subsec:imputation-algo}, which assumes that the entire gap should be filled in with flights, each of which covers the same fraction of the interval $[L^S_{I_j + N_j -1}; L^S_{I_{j+1}}]$ \citep{rhee2011levy, shin2007human}.

\item[unadjusted non-parametric method:] This method was originally proposed in \cite{barnett2020inferring}. Within this framework the distributions $q(\Delta^S_k|R_f = f, \bftheta)$ and $q(\Delta^S_k|R_k = p, \bftheta)$ are taken to be the weighted sample distributions of the flights and pauses that were recorded shortly before or after $L^T_k$, with increments closer in time having a greater weight. 
The authors also propose other approaches to estimating these distributions, all similar in their lack of adjustment for the MNAR sampling mechanism. A comprehensive comparison is beyond the scope of this paper.

\item[adjusted parametric method:] The plug-in imputation method described in Section \ref{sec:imputation} under the standard parameterization. This approach consists of estimating the model parameters using the observed increments and adjustment for the data collection model, then using the maximum likelihood estimates to impute missing parts of the trajectory.
\end{description}

We consider three data collection schemes:
\begin{description}
\item[\textbf{Unscheduled gap}:] Define the gap be the interval 
$$G_U(\alpha) = \left[ t_\text{max}/2 - \alpha t_\text{max}/2; t_\text{max}/2 - \alpha t_\text{max}/2\right]$$
and set $z_t = \mathbbm{1}(t \not\in G(\alpha))$. In this way, we mask the middle part of the trajectory which is of length $\alpha t_\text{max}$, where $\alpha=0.2 + l\frac{0.6}{49}$ for $l=0, 1, \dots 49$. This masking could correspond to a scenario in which a data collecting device malfunctions and does not record locations for a significant period of time (up to 80\% of the study period), or when the signal is lost due to the characteristics of the built environment (i.e. thick walls).
\item[\textbf{Unscheduled gap + short scheduled gaps}:] Recall the ``on-off" scheme described in Example \ref{ex:onoff} and set $o = u = 25$. The scheduled gaps are defined to be 
$$G_S(u, o) = \left\{t: t \in \left[B_t'; B_t'+I_u\right)\right\}.$$
We set $z_t = \mathbbm{1}(t \not\in G_U(\alpha) \cup G_S(u, o))$.
\item[\textbf{Unscheduled gap + long scheduled gaps}:] Similar to the previous case except that the scheduled gaps have length $o=u=50$.
\end{description}
All three schemes are illustrated in Figure \ref{fig:missing-data-schemes}:

\begin{figure}
    \centering
    \includegraphics[scale=0.6]{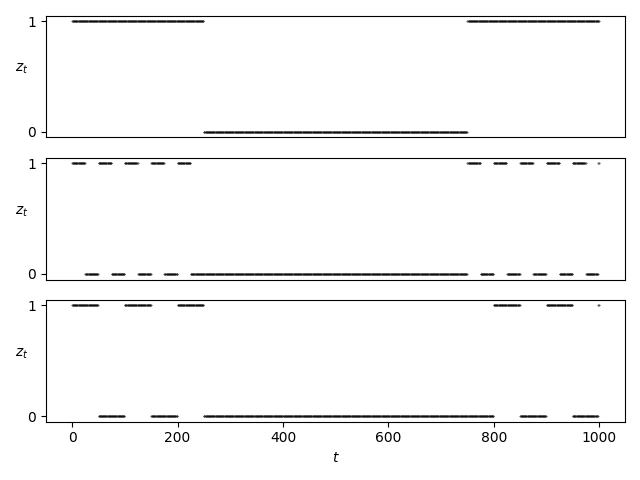}
    \caption{Data collection schemes considered in the numerical experiments. From top to bottom, the three panels illustrate an unscheduled gap and an unscheduled gap with short and long scheduled gaps, respectively.}
    \label{fig:missing-data-schemes}
\end{figure}

\begin{figure}[ht]
    \centering
    \includegraphics[width=1.0\textwidth]{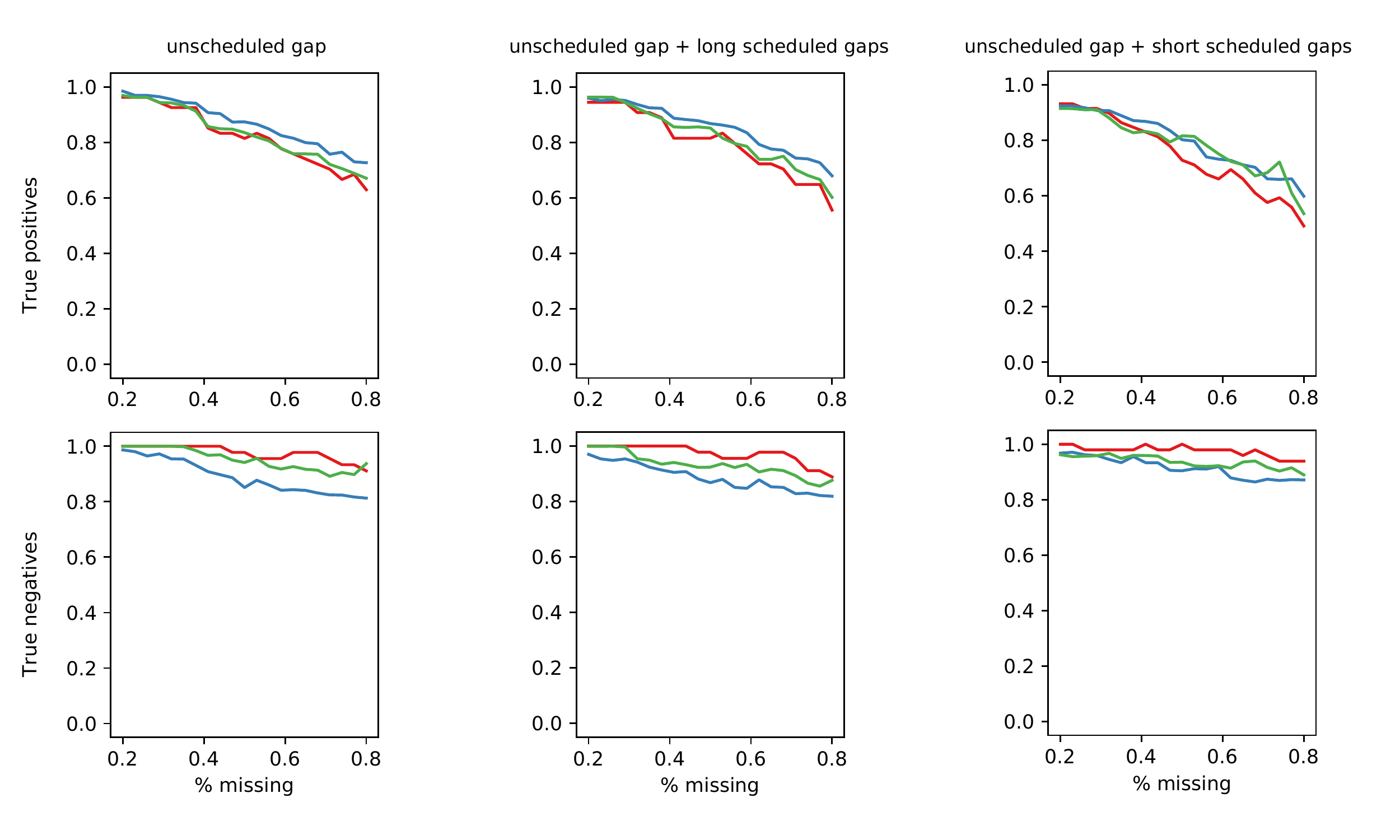}
    \caption{Mean probability of (not) passing through the hot-spot as a function of the percentage of the data which were missing and the type of the collection mechanism. Linear interpolation is marked in red (${\color{red}\textbf{-}}$), adjusted parametric method in blue (${\color{blue}\textbf{-}}$) and unadjusted nonoparametric in green (${\color{green}\textbf{-}}$)}
    \label{fig:simulation-tp-fp}
\end{figure}

We start by generating motions $\{\mathcal{m}_n\}_{n=1}^{N_{\text{tr}}}$ using the same parameter settings as in Section \ref{subsec:simulations-parameters}. Next for each realization of the motion we calculate the ``center of mass'' of its trajectory $\tau(\mathcal{m}_n)$ as $\bc = \frac{1}{t_{\max}}\sum_{\mathcal{D}(\mathcal{m}_n)} \bs^{n}_t$ and we shift all elements of $\tau(\mathcal{m}_n)$ by $-\bc$ which means that all simulated trajectories occupy roughly the same area. Examples of motions simulated in this way are shown in the supplementary materials. We then find the smallest bounding box $\mathcal{B} = [b^x_0, b^x_1] \times [b^y_0, b^y_1]$ such that locations from all trajectories $\tau(\mathcal{m}_1), \dots, \tau(\mathcal{m}_{N_{\text{tr}}})$ are within $\mathcal{B}$.

\begin{figure}[ht]
    \centering
    \includegraphics[width=1.0\textwidth, trim={0 0 0 1.3cm},clip]{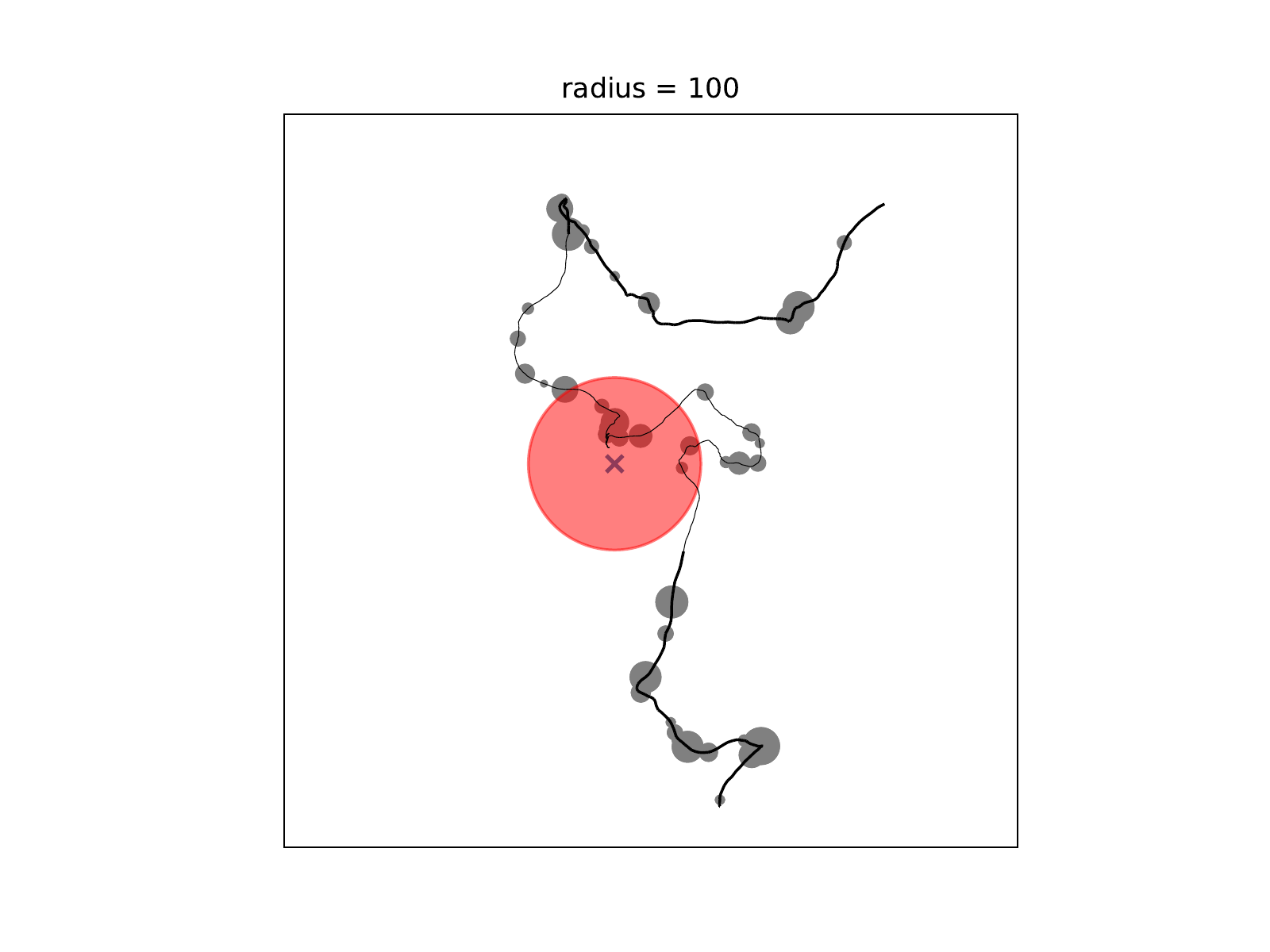}
    \caption{A sample simulated trajectory with a hotspot ($\boldsymbol{\times}$) and the exposure area of radius 100 marked in red. Missing data mechanism is of type ``unscheduled gap'' and the unobserved part of the realization of the motion is marked with a thinner line. The grey circles represent pauses with the center of the circle correponding to the location of the pause and its radius being proportional to its duration.}
    \label{fig:hotspot-example}
\end{figure}

In order to compare all the methods, motivated by the considerations presented in Section \ref{subsec:imputation-motivation} we focus on estimating exposure. To this end, we start by generating \emph{exposure hot-spots} as  $h_n = (h^x_n, h^y_n)$, where for $i \in \{x, y\}$ we have $h^i_n \sim \normal\left(0, \left(\frac{b^i_1 - b^i_0}{3}\right)^2\right)$ for $n = 1, \dots, N_{\text{tr}}$, and $\normal(\mu, \sigma^2)$ denotes a normal distribution with mean $\mu$ and variance $\sigma^2$. This way of selecting the hot-spot is intended to ensure that each trajectory is likely but not guaranteed to pass through the neighborhood of the hot-spot. We call this neighborhood  the \emph{exposure area} and define it as a ball $B(h_n, 100)$. Figure \ref{fig:hotspot-example} illustrates the concepts of a hot-spot and exposure area.
For each trajectory $\tau(\mathcal{m}_n)$, data collection scheme $\zeta$, hot-spot $h_n$ and missing percentage $\alpha$, we impute $N_\text{imp}=50$ times the missing portion of the trajectory using method $\nu$.

\paragraph{Exposure probability}

We start by evaluating the probability of passing through the exposure area. Let $W = \{n \in \{1, 2, \dots, N_{\text{tr}}\} : \tau(\mathcal{m}_n) \,\cap\, B(h_n, 100) \neq \emptyset\}$ be the set of of indices of trajectories which pass through the exposure area $B(h_n, 100)$ and let $V = \{1, 2, \dots, N_{\text{tr}}\} \setminus W$ be the indices of the remaining trajectories. If $n \in W$ then for each missing percentage $\alpha$, method $\nu$ and data collection scheme $\zeta$ we calculate $w_n^{\alpha, \nu, \zeta}$, the fraction of the trajectories imputed using $\nu$ that also pass through the danger zone (true positive). If $\tau(\mathcal{m}_n) \in V$ then calculate $v_n^{\alpha, \nu, \zeta}$, the fraction of the curves imputed using a given method that also do not pass through the exposure area (true negative).

We then compare average true positive $\bar{w}^{\alpha, \nu, \zeta} = \sum_{n \in W} w_n^{\alpha, \nu, \zeta}$ and true negative rates $\bar{v}^{\alpha, \nu, \zeta} = \sum_{n \in V} v_n^{\alpha, \nu, \zeta}$. A method is better if a given rate is higher. 

Each column in Figure \ref{fig:simulation-tp-fp} shows $\bar{w}^{\rho, \nu, \zeta}$ and $\bar{v}^{\rho, \nu, \zeta}$, while plots with $w_n^{\alpha, \nu, \zeta}$ and $v_n^{\alpha, \nu, \zeta}$ can be found in the supplementary material. 

Looking at the true positive averages we see that the method proposed in Section \ref{sec:imputation} outperforms the other two methods when there is no unscheduled missingness and when the scheduled gaps are large and is slightly better when the scheduled gaps are small. Note that shorter breaks are inherently easier to impute as the missing parts of the trajectory are more similar to a straight line and consequently all methods produce similar results. This suggests that a method which is better at imputing longer breaks should be preferred. In this case our method requires between 10 and 20 percentage point fewer observed locations than the other methods to achieve a comparable true positive rate.

The higher true positive rate exhibited by the adjusted parametric method is related to its lower true negative rate. In particular, the imputations generated using the adjusted parametric method explore the space more than the imputations generated using the other two methods (Section S1 in the supplement contains relevant examples). A more careful and application-dependent selection of the parametrization might to a  help to reduce the true negative rate while increasing the true positive rate. Moreover in certain applications one rate might be more important by the other.

Finally, it is important to note, that our simulations demonstrate that sampling schemes relying on short observation intervals degrade performance, as measured by our ``hot-spot metric". The true positive rates are roughly similar when there is no unscheduled missingness and when the scheduled observation sequences are longer. However when they are shorter, even if frequent, then the performance of the all the methods suffer.

\paragraph{Length of exposure}

A related evaluation relies on calculating the total time spent in the exposure area according to each imputation and comparing it with the amount of time spent in the area by the true trajectory.  

Consider trajectory $\tau(\mathcal{m}_n)$ and the corresponding imputation imputation $\omega_n^{\iota, \alpha, \nu, \zeta}$, where in addition to symbols used before we use $\iota=1, \dots, N_{\text{imp}}$ as the index of the imputation. Moreover, we consider only $\alpha\in \{0.2, 0.5, 0.8\}$. We define the true exposure time $e_n$ as the number of time periods during which the individual performing movement $\mathcal{m}_n$ is inside the exposure area, i.e.
$$
e_n = \#\{\mathcal{S}_t \in \tau(\mathcal{m}_n): |\mathcal{S}_t - h_n|^2_2 < 100\}.
$$
Similarly for each $\omega_n^{\iota, \alpha, \nu, \zeta}$ we write $e_n^{\iota, \alpha, \nu, \zeta} = \#\{\mathcal{S}_t \in \omega_n: |\mathcal{S}_t - h_n|^2_2 < 100\}$. We then calculate $d_n^{\iota, \alpha, \nu, \zeta} = e_n - e_n^{\iota, \alpha, \nu, \zeta}$, the difference between the true length of exposure and the length resulting from the imputed trajectory. For each imputation method and missing percentage we also calculate the grand mean of these differences as 
$$
\bar{d}^{\alpha, \nu, \zeta} = \frac{1}{N_\text{imp}N_\text{tr}}\sum_{n=1}^{N_\text{tr}}\sum_{\iota=1}^{N_\text{imp}}d_n^{\iota, \alpha, \nu, \zeta}.
$$
Figure \ref{fig:exposure-time} shows the distribution of $d_n^{\iota, \alpha, \nu, \zeta}$ for the "unscheduled gap" data collection scheme and the corresponding means. Analogous plots for the other two data collection schemes look similar and can be found in the supplementary material. 

\begin{figure}[ht]
  \centering
  \begin{subfigure}{0.48\textwidth}
    \includegraphics[width=1.0\textwidth]{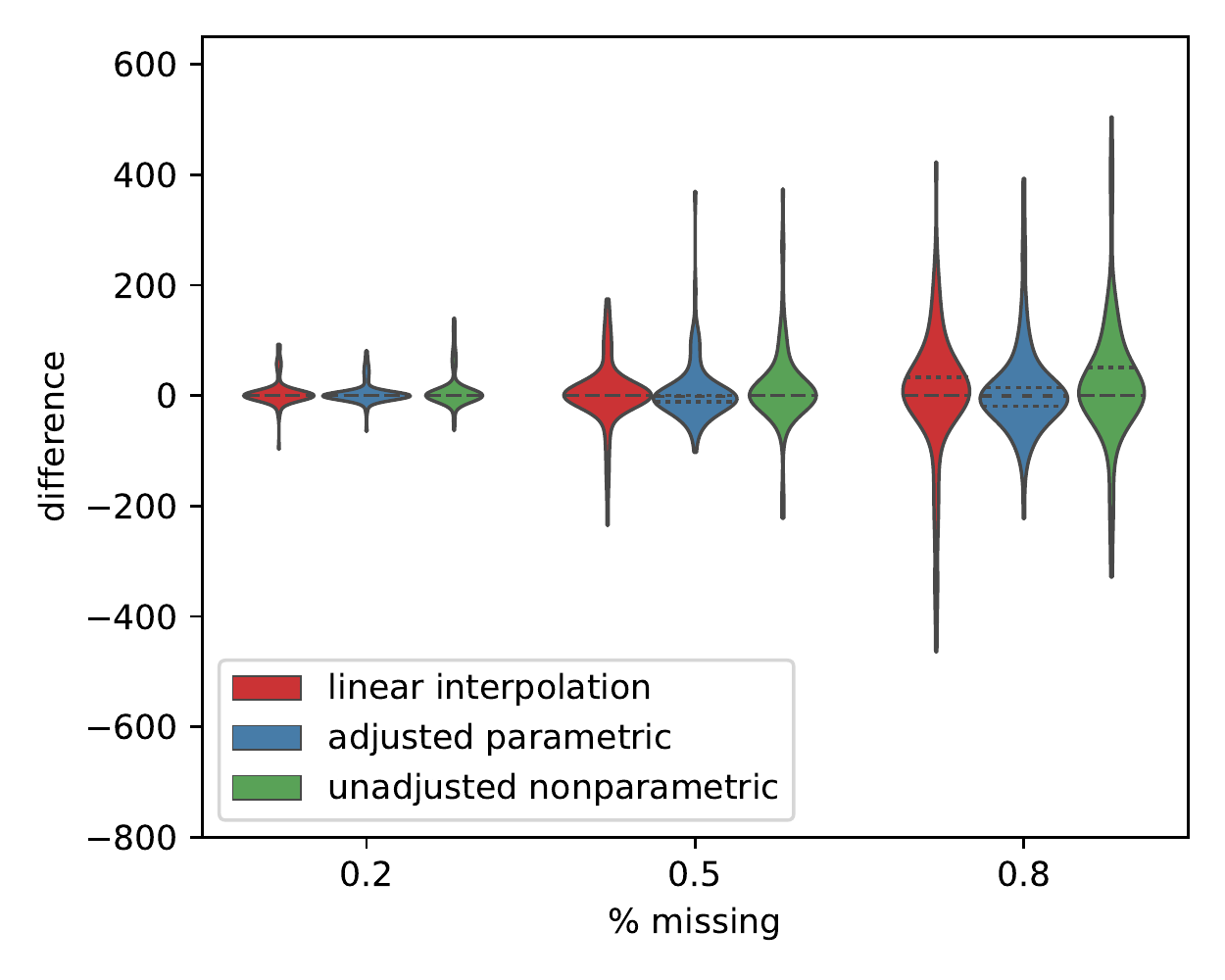}
    \caption{Unscheduled gap}
    \label{fig:exposure-just-hotspot}
  \end{subfigure}
  \begin{subfigure}{0.48\textwidth}
    \includegraphics[width=1.0\textwidth]{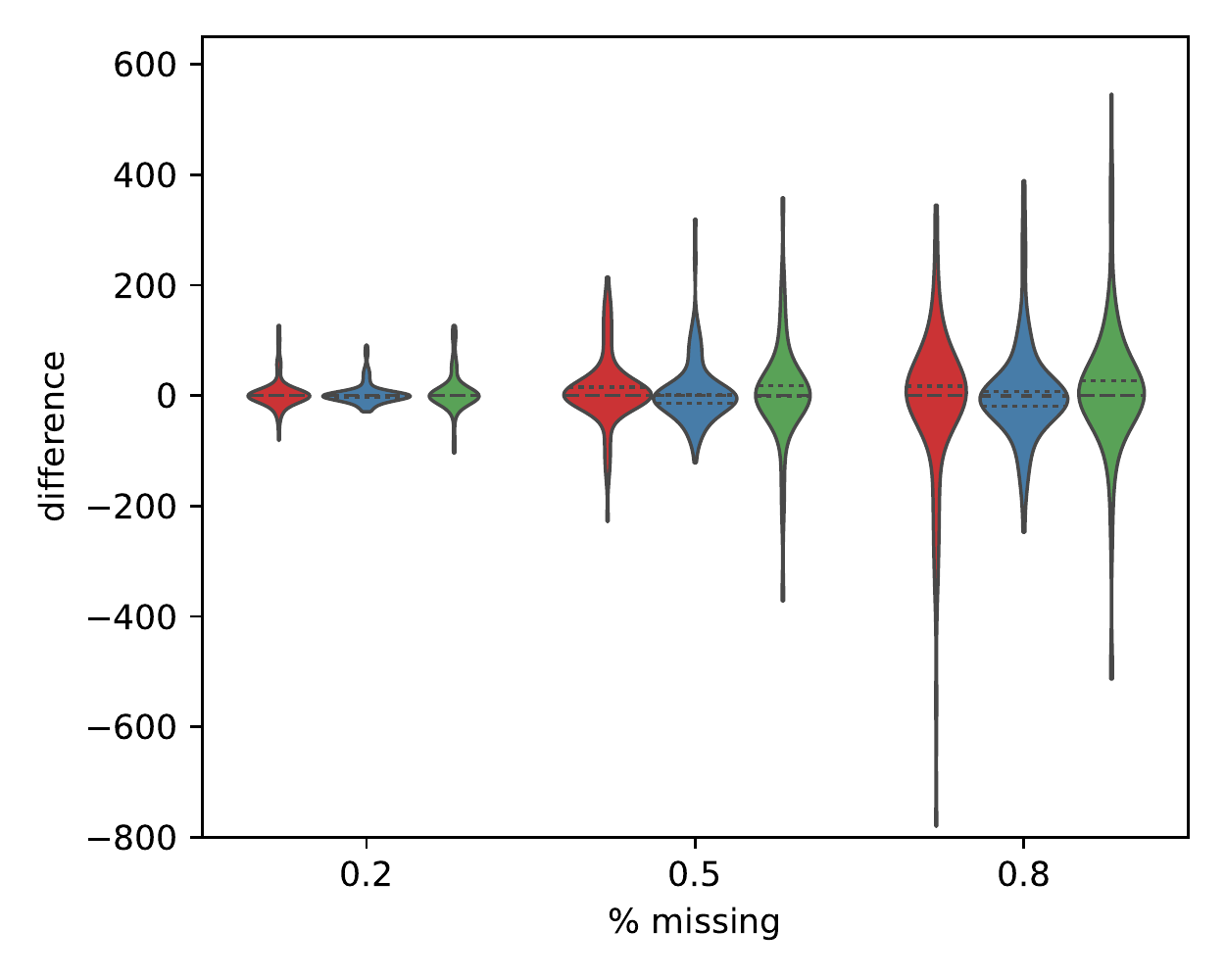}
    \caption{Unscheduled gap + long gaps}
    \label{fig:exposure-hotspot-long-obs}
  \end{subfigure}
  \begin{subfigure}{0.48\textwidth}
    \includegraphics[width=1.0\textwidth]{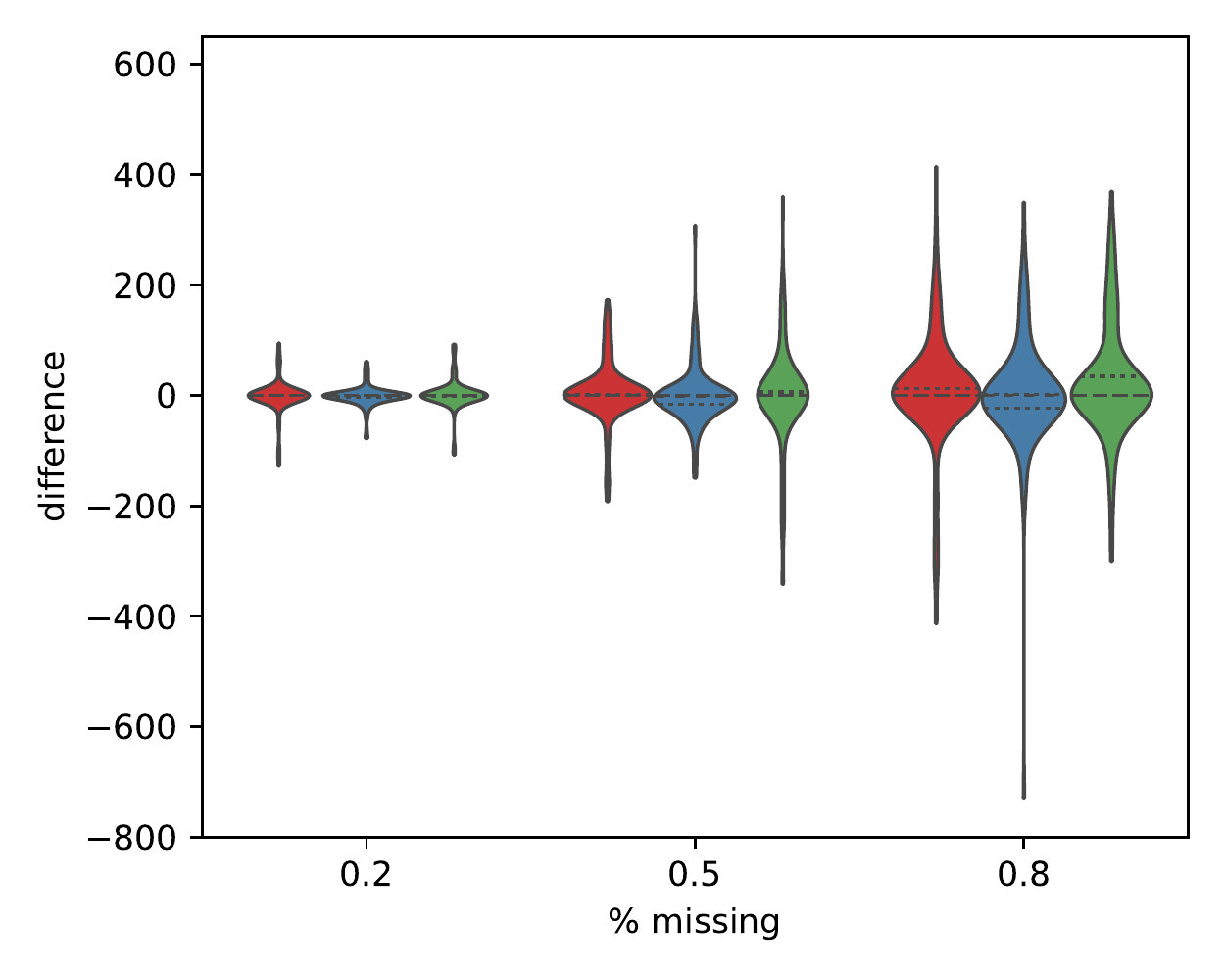}
    \caption{Unscheduled gap + short gaps}
    \label{fig:param-hotspot-short-obs}
  \end{subfigure}
  \caption{Distribution of differences between the true length of exposure and the length resulting from the imputed trajectory. The dashed line in the middle of the violin plot denotes the median while the dotted lines stand for the other two quartiles. In many cases some of the quartiles are almost identical.}
  \label{fig:exposure-time}
\end{figure}

Overall the adjusted parametric method we propose in our paper is somewhat better than the other two when there is no unscheduled missingness and when scheduled the gaps are long. This can be seen by the generally narrower spread of the distribution of differences. At the same time the distribution for all methods tend to have similar quartiles.

\subsection{Illustrative Application to Disney World Data}\label{subsec:data-application}

In this section we present an application of the framework developed in this paper to the analysis of real data. We use a collection of 41 trajectories (observed locations) corresponding to 19 individuals visiting the Disney World, near Orlando, FL \citep{ncsu-mobilitymodels-20090723}, some of whom came to the park several times.  The trajectories are made up of locations sampled every 30 seconds using a dedicated handheld GPS device. For illustration, we consider every trajectory (and the motion that it corresponds to) to be independent of all the others and governed by a unique set of parameters. Every location in the trajectory is represented using the $(x,y)$ coordinates which express the distance in meters from some predetermined reference point. 
Figure \ref{fig:application-trajectories} shows two examples of motions corresponding to the observed trajectories. Most of them consist of between 1000 and 2000 locations. There are no missing data.

We first compare the accuracy of parameter estimates using the parametric method with and without the missing data adjustment. To this end, we assume the standard parametrization and begin by estimating the parameters of the model using the entire motion $\mathcal{m}_n$ for $n=1,\dots, 41$. We call these estimates $\hat{\bftheta}^{n,t}$, representing the value of the MLE for the complete data. Next, we remove some locations following the ``unscheduled gap + short scheduled gaps" mechanism, described in Section \ref{subsec:simulations-imputation}, and assuming that $\alpha = 0.5$. We then estimate parameters $\hat{\bftheta}^{n,b}$ assuming that increments derived from the observed locations are missing at random, which results in estimates that are biased relative to the estimates based on the entire observed trajectory. As the last step we use the form of observed data likelihood given in Proposition \ref{thm:observed-data-likelihood} which accounts for the data collection mechanism to obtain estimates of $\hat{\bftheta}^{n,u}$.
Then, for motion $\mathcal{m}_n$, we calculate $D^{n,j}$, the relative difference between the estimate obtained using the entire motion $\mathcal{m}_n$ and the estimate obtained using this motion with missing (masked) increments. Mathematically,
$$
D_i^{n,j} = \frac{\hat{\theta}_i^{n,j} - \hat{\theta}_i^{n,t}}{\hat{\theta}_i^{n,t}},
$$
where $j \in \{b, u\}$ and $i = 1, 2, 3, 4$ indicates the component of $\hat{\bftheta}^{n, j}$. The histograms of these relative differences are shown in Figure \ref{fig:application-params}.

Our results show that without the adjustment for the data collection mechanism, the parameters often deviate significantly (even in relative terms) from the estimates obtained using the entire motion. 

Next, using the methods in Section \ref{sec:simulations} and masking the locations to mimic the short scheduled gaps mechanism we calculate the probability of each of the individuals passing through a randomly selected hot-spot. Since the total number of trajectories is fairly low, unlike in Section \ref{sec:simulations}, we generated 20 hot-spots for each trajectory and calculated the probability of passing through each one of them. The results are reported in Figure \ref{fig:application-results} using the same format as in Section \ref{sec:simulations}. 
Similar to the results obtained in Section \ref{subsec:simulations-imputation}, here we also see that the adjusted parametric method results in fewer false negatives than the other two as captured by the hot-spot metric described in Section \ref{sec:simulations}.
\begin{figure}[ht]
    \centering
    \begin{subfigure}[b]{0.68\textwidth}
         \includegraphics[width=\textwidth]{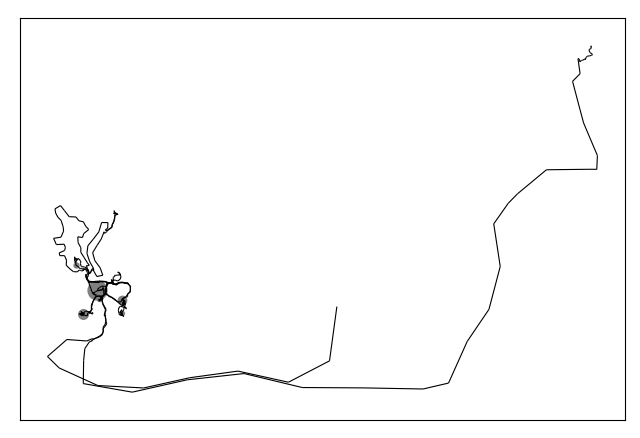}
    \end{subfigure}
    \hfill 
    \begin{subfigure}[b]{0.29\textwidth}
         \includegraphics[width=\textwidth]{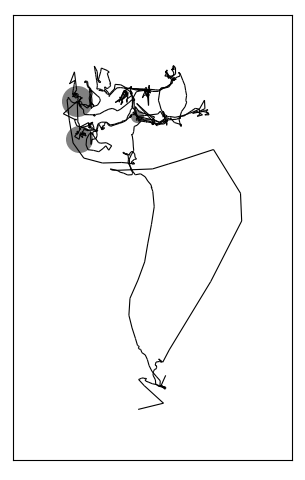}
    \end{subfigure}
    \caption{Sample trajectories from the analyzed data set.}
    \label{fig:application-trajectories}
\end{figure}

\begin{figure}[ht]
    \centering
    \includegraphics[width=1.0\textwidth]{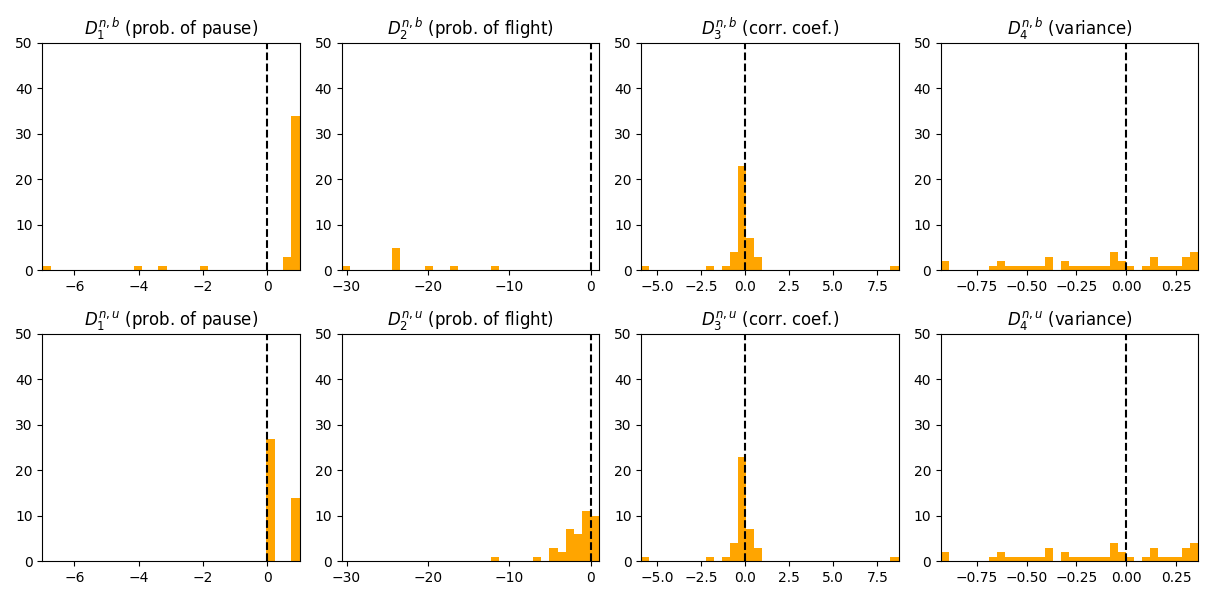}
    \caption{Histograms or relative differences in parameter estimates. The first row shows the differences for the estimates obtained under the MAR assumption while the second row shows the differences for the estimates obtained after adjusting for the data collection mechanism. The differences in parameters related to flights do not change between the rows because, like in Section \ref{sec:simulations}, the direction and length of the flights are in fact missing at random. The histogram of $D_2^{n,b}$ has only a few values because for trajectories in which there are no observed pauses this parameter cannot be estimated.}
    \label{fig:application-params}
\end{figure}

\begin{figure}[ht]
    \centering
    \includegraphics[width=1.0\textwidth]{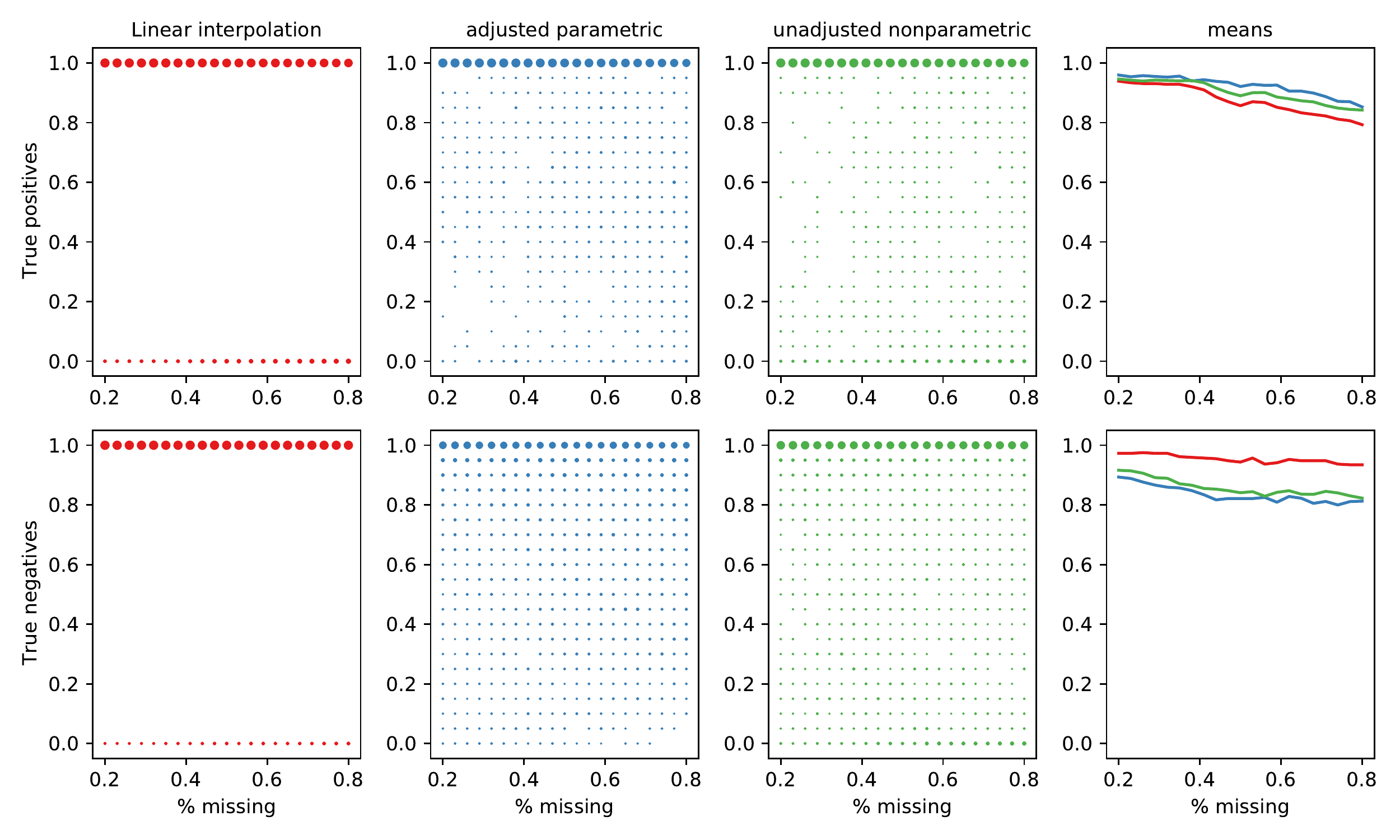}
    \caption{Probability of (not) passing through the hot-spot as a function of the length of the data which were missing outside of the schedule.}
    \label{fig:application-results}
\end{figure}

\section{Key takeaways and discussion}\label{sec:discussion}

In this paper, we provide a rigorous statistical formulation of the flight-pause model to represent individual human mobility. We derived the mathematical representation of the FPM by formalizing the notion of increments, decomposing their joint likelihood and by providing an explicit connection to observed locations. Our model builds on previous work describing mobility as a series of flights and pauses but, unlike previous work in this field, our approach leads to a likelihood function. Thus it can serve as a basis for various extensions to statistical inference, including in particular formal links to notions of data collection mechanisms, assumptions about them, and possible missing data adjustments.  We have shown how our formulation can lead to both methods for inference and movement imputation that can improve upon biases exhibited from previous approaches that have operated, possibly erroneously, under the assumption that increments making up data on a motion were missing at random.

The insights gained from the formulation of motions made up of increments and increments' relationship to locations introduce several implications for the design of studies using MPT.  For example, when the goal is to infer features of movement with the FPM, researchers using this framework should design measurement schemes in terms of their implications for increments and not, as is commonly done, in terms of their impact on locations. Some of these implications can be measured using the concept of the effective sample size, which we introduced. All of the above considerations must ultimately be balanced against practical restrictions (e.g., relating to device battery), but we hope that the ideas presented in this paper provide valuable tools to account for these restrictions in a way compatible with the goals of a particular investigation.

Our approach is not without limitations. First and foremost, expressing the motion as a sequence of increments introduces much complexity in the analysis of missing data, on which we elaborated extensively in Section 4. A more basic model, for example, one which assumes that the locations and not increments form a Markov chain, might make this analysis simpler. At the same time, our approach is motivated by the paradigm used in several prominent works in the area of human mobility and it also allows for explicit modeling of the duration of pauses - a simple Markov chain cannot do that.

Another limitation of our model lies in its inability to make use of some data if it is collected in a way radically different than the data collection schemes observed here. More specifically, observed locations separated from other observed locations in time, cannot be classified as increments and are thus not used in the construction of the FPM likelihood or in the proposed trajectory interpolation. For example, an isolated location observed in the middle of a large gap would not be used for parameter inference, and corresponding trajectory interpolations from the model are not guaranteed to pass through that location. Extensions to incorporate such ``singleton" observations may be possible, for example, through additional restrictions to the likelihood or as auxiliary information used to interpolate trajectories.

We also acknowledge the simplicity of the standard parameterization and note that an alternative formulation would be needed before the model could be useful in analyzing complex MPT data for a specific purpose. Moreover, while we assumed that the locations are observed without error, this is typically not the case. We thus leave for future work the consideration of the implications of inaccurate observations. Furthermore, we see a potential for improvement in modeling the motion of several individuals jointly (such as in \citet{scharf2016dynamic} or \citet{milner2021modelling}) and accounting for their possible interactions. We view these imperfections of our method as promising directions for future research.

Finally, it would be interesting to explore the connections between our work and the moving-resting process (MR). It is possible that a certain parametrization of our model could be considered a discretized approximation of the MR model. Similarly, there are clear connections between our FPM and the semi-Markov model specified on locations in \citet{langrock2012flexible}, and it would be worthwhile to explore connections between these approaches, particularly as they relate to inference under missing data in common human mobility study designs.


\footnotesize

\appendix
\section*{Acknowledgments and Funding}
This study was supported in part by the Eunice Kennedy Shriver National Institute on Child Health and Human Development (Catherine A. Calder, R01HD088545; Elizabeth Gershoff, The University of Texas at Austin Population Research Center, P2CHD-042849) and by the National Institutes of Health grant NIH R01ES026217. We would also like to thank Justin Drake, Raymond Wang and Nathan Wikle for helpful discussions and comments. Special thanks to Giovanni Rebaudo for his valuable comments on the technical aspects of the paper.

\FloatBarrier

\section*{Conflicts of interests}
All authors declare that they have no conflicts of interest.

\section*{Data availability}
Data sharing not applicable to this article as no datasets were generated during the current study

\section{Derivation of the state-space model}\label{app:state-space-model}

In Section \ref{sec:imputation} we used sampling from the smoothing distribution (FFBS) of the direction of flights to impute a missing part of the trajectory. Our approach can be detailed as follows. Consider $k_l$ to be a sequence of indices such that $R_{k_l} = f$. Then we have 
$$\label{eq:delta-as-function-of-L}
\Delta^S_{k_l} = L^S_{k_{l+1}} - L^S_{k_l}.
$$
Under standard parametrization this means that 
$$
\Delta^S_{k_l} = \theta_3 \Delta^S_{k_{l-1}} + \bfepsilon_l,
$$
where $\epsilon_l \sim \mathcal{N}(0, \theta_4\bI)$. Because of \eqref{eq:delta-as-function-of-L} we can also write it as
$$
L^S_{k_{l+1}} - L^S_{k_l} = \theta_3(L^S_{k_l} - L^S_{k_{l-1}}) + \bfepsilon_l.
$$
Thus if we now define $\bx_l = \left[L^S_{k_l}, \Delta^S_{k_{l-1}}\right]$, we can write that 
$$
\bx_{k_l} = \bA(\bftheta)\bx_{k_{l-1}} + \bv_{k_l},
$$
where $\bv_{k_l} \sim \normal(\bfzero, \bQ(\bftheta))$ with
\begin{equation*}
    \bA(\bftheta) = \left[ \begin{array}{cccc} 
                    1 & \theta_1 & 0 & 0 \\
                    0 & 0 & 1 & \theta_1 \\
                    0 & \theta_1 & 0 & 0 \\
                    0 & 0 & 0 & \theta_1
                \end{array}\right] 
    \quad\quad \text{and} \quad \quad
    \bQ(\bftheta) = \theta_2\left[ \begin{array}{cccc} 
                    1 & 1 & 0 & 0 \\
                    0 & 0 & 1 & 1 \\
                    1 & 1 & 0 & 0 \\ 
                    0 & 0 & 1 & 1
                \end{array}\right].
\end{equation*}

\section{Inference under the standard parametrization}

In this section show how some of the general results derived earlier in the paper can be made specific to the case of standard parametrization. In addition, throughout this Section we adopt
\begin{assumption}\label{ass:indep-nablas}
The direction and distance of the first flight in an observed block $O_j$ do not depend on the direction and distance of the last increment in the preceding block, $U_{j-1}$. Formally, if $k = I_j$ then
$$
p\left(\Delta^S_{I_j}|\bM_{I_j - 1}, R_k = f, \bftheta\right) = p\left(\Delta^S_{I_j}|R_k = f, \bftheta\right).
$$
\end{assumption}
Similar though less restrictive than Assumption \ref{ass:indep-blocks}, this assumption can be a good approximation of the truth if the blocks of unobserved increments are long but might lead to stronger bias when most of them are very short.

\subsection{Observed data likelihood}\label{app:inference-obs-data}

In this Section we substitute the parametric families from Section \ref{subsec:standard-param} into expression \eqref{eq:obs-data-integral}. We define the following sets of indices:
\begin{align*}
P_j&= \{k : \bM_k \in \mathcal{P}\cap O_j\} ,\quad & F_j &= \{k : \bM_k \in \mathcal{F} \cap O_j \setminus \{\bM_{I_j}\}\}, & F^f_j &= \{k : \bM_k \in \mathcal{F}^f \cap O_j \setminus \{\bM_{I_j}\}\}, \\
\tilde{P}_j &= \{k : \bM_k \in \mathcal{P}\cap U_{j-1}\} ,\quad & \tilde{F}_j &= \{k : \bM_k \in \mathcal{F} \cap \left(\{\bM_{I_j}\} \cup U_{j-1}\right)\}, & \tilde{F}^f_j &= \{k : \bM_k \in \mathcal{F}^f \cap \left(\{\bM_{I_j}\} \cup U_{j-1}\}\right).\\
\end{align*}
Using this notation \eqref{eq:obs-data-integral} can be transformed into
\begin{multline}\label{eq:obs-data-integral-standard}
    q(\mathcal{O}|\bfxi, \bftheta) = 
    \prod_{j=1}^J \frac{\theta_1^{|P_j|} (1-\theta_1)^{|F^f_j|}
    \cdot \theta_2^{|P_j|}  \left(1-\theta_2\right)^{\sum_{k \in P_j}\Delta^T_k}}{\left(2\pi\theta_4\right)^{\frac{|F_j|}{2}}} \cdot \exp\left(-\frac{1}{2\theta_4}\sum_{k\in F_j} \left(\Delta^S_k - \theta_3\Delta^S_{k - 1}\right)^2\right) \cdot \\
    \cdot q(\bM_{I_1}| \bftheta) \cdot \bigintss  q(\bfxi|\mathcal{O}, \mathcal{U}, \bfpsi) \prod_{j=2}^{J-1} \frac{\theta_1^{|\tilde{P}_j|} (1-\theta_1)^{|\tilde{F}^f_j|} \cdot \theta_2^{|\tilde{P}_j|} \left(1-\theta_2\right)^{\sum_{k \in \tilde{P}_j}\Delta^T_k}}{{\left(2\pi\theta_4\right)^{\frac{|\tilde{F}_j|}{2}}}} \cdot \\
    \cdot \exp\left(-\frac{1}{2\theta_4}\sum_{k \in \tilde{F_j}} \left(\Delta^S_k - \theta_3\Delta^S_{k-1}\right)^2\right) d\mathcal{U}.
\end{multline}
If we use $I(\bfxi, \mathcal{O}, \bfpsi, \bftheta)$ to denote the integral, then the logarithm of the observed data likelihood takes the form
\begin{multline}\label{eq:obs-data-loglik-std}
    \log q(\mathcal{O}|\bfxi, \bftheta) = \sum_{j=1}^J |P_j|\log \theta_1 + |F^f_j| \log (1-\theta_1) + |P_j| \log \theta_2 + \sum_{k \in P_j} (\Delta^T_k - 1) \log(1-\theta_2) - \frac{F_j}{2}\log (2\pi\theta_4) - \\
    - \frac{1}{2\theta_4}\sum_{k \in F_j} \left(\Delta^S_k - \theta_3\Delta^S_{k-1}\right)^2 + \log q(\bM_{I_1}|\bftheta) + \\
    + \log \bigintss  q(\bfxi|\mathcal{O}, \mathcal{U}, \bfpsi) \prod_{j=2}^{J-1} \frac{\theta_1^{|\tilde{P}_j|} (1-\theta_1)^{|\tilde{F}^f_j|} \cdot \theta_2^{|\tilde{P}_j|} \left(1-\theta_2\right)^{\sum_{k \in \tilde{P}_j}\Delta^T_k}}{{\left(2\pi\theta_4\right)^{\frac{|\tilde{F}_j|}{2}}}} 
    \cdot \exp\left(-\frac{1}{2\theta_4}\sum_{k \in \tilde{F_j}} \left(\Delta^S_k - \theta_3\Delta^S_{k-1}\right)^2\right) d\mathcal{U}
\end{multline}

\subsection{Simplifications for select MNAR mechanisms}\label{app:inference-mnar}

Recall that in Section \ref{subsec:observed-data-likelihood} we observed that that observed increments can be partitioned into blocks $O_1, \dots, O_j$. In a similar manner we now partition the elements of the motion trajectory $\tau(\mathcal{M})$ into blocks of observed and unobserved locations, which we denote with $\tilde{O}_j$ and $\tilde{U}_j$, respectively. Unlike in the case of increments, however, for a block of observations to be declared "observed" it has to contain at least one pair of locations whose spatial coordinates are not the same. We also note that $\tilde{O}_j \supset \tau(O_j)$ and $\tilde{U}_j \subset \tau(U_j)$.

We use $\tilde{T}_j$ to stand for the time of the first location in block $\tilde{O}_j$ and let $\tilde{N}_j$ be its length and define the following variables
\begin{align*}
    d_j &= \underbrace{(\tilde{T}_j + \tilde{N}_j - 1)}_{\text{time of the last observation in } \tilde{O}_j} - \underbrace{L^T_{I_j + N_j - 1} + \Delta^T_{I_j + N_j - 1}}_{\text{the last time point in } \mathcal{D}(O_j)} \\
    g_j &= \underbrace{L^T_{I_j}}_{\text{first time point in }\tilde{O}_j} - \underbrace{\tilde{T}_j}_{\text{first time point in }\mathcal{D}(O_j)}
\end{align*}
Intuitively, $d_j$ is the number of locations following $O_j$ which are observed but do not consitute an observed increment and $g_j$ is the number of locations which directly precede $O_j$ but also do not make up an observed increment. Note that when the last two locations in $\tilde{I}_j$ are different, i.e. when $\mathcal{S}_{\tilde{T}_j + \tilde{N}_j - 1} \neq \mathcal{S}_{\tilde{T}_j + \tilde{N}_j - 2}$ then $d_j = 0$. Analogously, $g_j = 0$ if the first two elements in $\tilde{I}_j$ are different, i.e., when $\mathcal{S}_{\tilde{T}_j} \neq \mathcal{S}_{\tilde{T}_j + 1}$. To distinguish between these situations we define $\delta_j = \mathbbm{1}(d_j = 0)$ and $\gamma_j = \mathbbm{1}(g_j = 0)$.
We also need to define
\begin{align*}
    j(t) &= \max \left\{ j : \tilde{T}_j \leq t\right\} \\
    B_t &= \tilde{T}_{j(t)}\\
    B_t' &= \tilde{T}_{j(t)} + \tilde{N}_{j(t)}.
\end{align*}
In words, $j(t)$ is a mapping that for each time $t$ assigns the index of the most recent block of observed location, while $B_t$ and $B_t'$ denote the time of the first location in, correspondingly, the blocks of observed and unobserved locations.
We can then provide a tractable expression for the log-likelihood for the data collection scheme described in Example \ref{ex:onoff}.

\begin{proof}[Proof of Proposition \ref{thm:observed-data-likelihood}]
Let us start by rearranging the terms in \eqref{eq:log-obs-data-standard}. We have
\begin{multline}
    \log q(\mathcal{O}|\bftheta) = \sum_{j=1}^J|P_j| \log\theta_1 + |F_j| \log(1-\theta_1) + |P_j| \log\theta_2 + \sum_{k \in P_j}(\Delta^T_k  - 1 )\log\left(1-\theta_2\right) - \frac{|\mathcal{F}\cap O_j| - 1}{2}\log \left(2\pi\theta_4\right) \\
    - \sum_{j=1}^J \frac{1}{2\theta_4} \sum_{k = I_j + 1} \left(\Delta^S_k - \theta_3\Delta^S_{k-1}\right)^2 + \log p\left(\Delta^S_{I_1}|\bftheta\right) + \\
    + \sum_{j=1}^J \log q\left(\Delta^S_{I_j}|\bftheta\right) + \log\theta_1 \sum_{j=1}^{J-1} \delta_j +  \log\theta_2\sum_{j=1}^{J-1}\gamma_j + \sum_{j=1}^{J-1} \log P_{\delta_j+1, \gamma_{j+1}+ 1}(N_j) + \log(1-\theta_2) \sum_{j=1}^{J-1}d_j + g_{j+1}.
\end{multline}
We can see that the first two lines of the expression above match the first two lines of \eqref{eq:obs-data-loglik-std}. It remains to show that the last line above is equal to the last line of \eqref{eq:obs-data-loglik-std}.

First note that $q(\bfxi|\mathcal{M}, \bfpsi) = \prod_{k=1}^K q(\xi_k|\bM_k, \bfpsi)$ because whether a particular increment is observed depends only on at what time point it starts and ends. More specifically 
\begin{equation}\label{eq:p_xi_m_given_m}
    q(\xi_k|\bM_k) = \begin{cases}
      1, & \begin{array}{cc}
           \xi_k = 1, & \left(L^T_k < B_{L^T_k}'\right) \land \left(\left(\left(\Delta^T_k < B_{L^T_k}' - L^T_k\right) \land R_k = p\right) \lor \left(R_k=f\right)\right)\\
           \xi_k = 0, & \text{otherwise} 
      \end{array}\\
      0, & \text{otherwise}
    \end{cases}.
\end{equation}
Using indicator notation we can also write it as
\begin{multline}\label{eq:inc-obs-indic}
q(\xi_k|\bM_k) = \mathbbm{1}(\xi_k=1)\mathbbm{1}\left(B_{L^T_k}' > L^T_k\right)\left( \mathbbm{1}(R_k=f) + \mathbbm{1}(R_k=p)\mathbbm{1}\left(\Delta_k^T < B_{L^T_k}' - L^T_k\right)\right) + \\
+\mathbbm{1}(\xi_k=0)\left(\mathbbm{1}(L^T_k \geq B_{L^T_k}') + \mathbbm{1}(L^T_k < B_{L^T_k}')\mathbbm{1}(R_k=p)\mathbbm{1}(\Delta^T_k \geq B_{L^T_k}' - L^T_k)\right). 
\end{multline}

This observation allows us to decompose the integral as
\begin{multline}
     \sum_{j=1}^{J-1} \log \int \prod_{k = I_j + N_j}^{I_{j+1}} q(\xi_k|\bM_k, \bfpsi) \cdot \frac{\theta_1^{|\tilde{P}_j|} (1-\theta_1)^{|\tilde{F}^f_j|} \cdot \theta_2^{|\tilde{P}_j|} \left(1-\theta_2\right)^{\sum_{k \in \tilde{P}_j}\Delta^T_k}}{{\left(2\pi\theta_4\right)^{\frac{|\tilde{F}_j|}{2}}}} \cdot \\
    \cdot \exp\left(-\frac{1}{2\theta_4}\sum_{k \in \tilde{F}_j} \left(\Delta^S_k - \theta_3\Delta^S_{k-1}\right)^2\right) dU_j = (\bullet).
\end{multline}
Next, using Assumption \ref{ass:indep-nablas} we can further simplify $(\bullet)$ into
\begin{multline}
    (\bullet) = \sum_{j=2}^J \log q\left(\Delta^S_{I_j}\right) + \\
    + \sum_{j=1}^{J-1} \log \int \prod_{k = I_j + N_j}^{I_{j+1}} q(\xi_k|\bM_k, \bfpsi) \cdot \underbrace{\theta_1^{|\tilde{P}_j|} (1-\theta_1)^{|\tilde{F}^f_j|} \cdot \theta_2^{|\tilde{P}_j|} \left(1-\theta_2\right)^{\sum_{k \in \tilde{P}_j}\Delta^T_k}}_{(\star\star)} d\{(R_k, \Delta^T_k : \bM_k \in U_j\} = (\bullet\bullet).
\end{multline}
Using $\mathcal{I}_j$ to denote the integral under the summation we can write 
$$(\bullet\bullet) = \sum_{j=2}^J \log q\left(\Delta^S_{I_j}\right) + \sum_{j=1}^{J-1} \log \mathcal{I}_j.$$
Next, if we define $q_t = \mathbbm{1}(t \in \mathcal{D}(\mathcal{P}))$, then it turns out that that $\{q_t, t\in\mathbb{N}\}$ is a Markov chain with transition matrix 
$$
\mathcal{Q} = \left[\begin{array}{cc}1-\theta_1 & \theta_1 \\ \theta_2 & 1 - \theta_2\end{array}\right].
$$
For a two-state Markov chain there exist explicit formulas which allow us to calculate the $n-$step transition matrix, i.e. the $n$-th power of the $Q$ matrix. Specifically, it can be shown by induction that 
\begin{equation}\label{eq:Qn-def}
Q(n) = \mathcal{Q}^n = \frac{1}{\theta_1 + \theta_2}\left(\left[\begin{array}{cc}\theta_2 & \theta_1 \\ \theta_2 & \theta_1\end{array}\right] + (1 - \theta_1 - \theta_2)^n \left[\begin{array}{cc}\theta_1 & -\theta_1 \\ -\theta_2 & \theta_2\end{array}\right]\right).
\end{equation}
Now notice that $\prod_{k = I_j + N_j}^{I_{j+1}} q(\xi_k|\bM_k, \bfpsi) = q(\xi_{\bM_{I_{j+1}}} = 1|\bM_{I_{j+1}}, \bfpsi)\prod_{k = I_j + N_j}^{I_{j+1} - 1} q(\xi_k=0|\bM_k, \bfpsi)$. 
Combining this with \eqref{eq:inc-obs-indic} and \eqref{eq:Qn-def} we have that $\mathcal{I}_j = \theta_1^{\delta_j}(1-\theta_2)^{d_j} Q_{\delta_j, \gamma_{j+1}}(|U_j|)(1-\theta_2)^{g_{j+1}}\theta_2^{\gamma_{j+1}}$. Taking logs of this expression completes the proof. 
\end{proof}

\section{Other Proofs}

\subsection{Results from Section \ref{sec:model}}

\begin{proof}[Proof of Proposition \ref{thm:complete-data-likelihood}]
Let us start with the joint probability distribution of the motion, parametrized by $\bftheta$ which takes the form
$$
q(\mathcal{M}|\bftheta) = q(\bM_1|\bftheta) \prod_{k = 1}^K q(\bM_k|\bM_{k-1}, \dots, \bM_1, \bftheta) = q(\bM_1|\bftheta) \prod_{k = 2}^K q(\bM_k|\bM_{k-1}, \bftheta).
$$
where the last equality is due to Assumption \ref{ass:Markov}. Focusing on the individual term under the product sign and using the definitions introduced at the beginning of Section \ref{sec:model} we can express it as
$$
q(\bM_k|\bM_{k-1}, \bftheta) = q(L_k, \Delta_k, R_k|\bM_{k-1}, \bftheta) = q(L_k, \Delta_k||\bM_{k-1}, R_k, \bftheta)q(R_k|\bM_{k-1}, \bftheta).
$$
If we group together the terms $q(R_k|\bM_{k-1}, \bftheta)$ for all $k$s, it remains to be proven that 
$$
\prod_{R_k = f, k>1} \cdot q(\Delta^T_k|R_k, \bM_{k-1}, \bftheta) \cdot \prod_{R_k = p, k>1} q(\Delta^S_k|R_k = p, \bM_{k-1}, \bftheta) = \prod_{k=2}^K q(L_k, \Delta_k||\bM_{k-1}, R_k, \bftheta)
$$
Using the definition of $\Delta_k$ and $L_k$ we rewrite $q(L_k, \Delta_k|\bM_k, R_k, \bftheta)$ as
\begin{multline}
q(L^T_k, L^S_k, \Delta^S_k, \Delta^T_k|\bM_k, R_k, \bftheta) = \\ \mathbbm{1}(R_k = f)\mathbbm{1}(\Delta^T_k =  1)\mathbbm{1}(L^T_k = L^T_{k-1} + \Delta^T_{k-1}) \cdot \\ 
\cdot \left( \mathbbm{1}(L^S_k = L^S_{k-1} + \Delta^S_{k-1})\mathbbm{1}(R_{k-1}=f) + \mathbbm{1}(L^S_k = L^S_{k-1})\mathbbm{1}(R_{k-1}=p)\right) \cdot q(\Delta^S_k|\bM_{k-1}, R_k, \bftheta) + \\
+ \mathbbm{1}(R_k = p)\mathbbm{1}(L^S_k = L^S_{k-1} + \Delta^S_{k-1})\mathbbm{1}(\Delta^S_k=\Delta^S_{k-1}) \cdot q(\Delta^T_k|\bM_{k-1}, R_k, \bftheta)
\end{multline}
The indicator functions \eqref{ass:continuity-in-time}-\eqref{ass:continuity-in-space-pauses}, \eqref{ass:flights-have-length-1} and \ref{ass:DeltaS}. In the last line we also used Assumption \ref{ass:no-cons-pauses}. Omitting indicator functions for clarity of notation we finish the proof.
\end{proof}

\begin{proof}[Proof of Proposition \ref{thm:standard-param-likelihood}]
We use specifications \ref{spec:four-dim}-\ref{spec:gaussian-increments} to substitute for relevant components of \eqref{eq:complete-data-likelihood}. 

First, using specification \ref{spec:gaussian-increments} we see that $\prod_{R_k = p, k>1} q(\Delta^S_k|R_k = p, \bM_{k-1}, \bftheta)$ becomes 
$$
\prod_{k > 1} \prod_{i=1, 2} \normal((\Delta^S_k)_i; \theta_3(\Delta^S_{k-1})_i, \theta_4).
$$

Next, if we define $y_t = \mathbbm{1}(t \in \mathcal{D}(\mathcal{P}))$, then it turns out that that $\{y_t, t\in\mathbb{N}\}$ is a Markov chain with transition matrix 
$$
\mathcal{Q} = \left[\begin{array}{cc}1-\theta_1 & \theta_1 \\ \theta_2 & 1 - \theta_2\end{array}\right],
$$
because of specification \ref{spec:markov-in-types}, \ref{spec:constant-pause-prob} and \ref{spec:geom}. This means that we can equivalently represent 
$$
\prod_{R_k = f, k>1} q(\Delta^T_k|R_k, \bM_{k-1}, \bftheta)
$$
as
$$
 \theta_1^{|\mathcal{P}|} (1-\theta_1)^{|\mathcal{F}^f|} \cdot \theta_2^{|\mathcal{P}|}  \left(1-\theta_2\right)^{\sum_{k : R_k=p}\Delta^T_k}.
$$
\end{proof}

\subsection{Results from Section \ref{sec:data-model}}

In the remainder of this section we will make use of the following definitions. 
\begin{definition}\label{def:anchor-locations}
Let $\{\bs_t\}_{t \in \mathbb{N}}$ be a realization of trajectory $\{\mathcal{S}_t\}_{t \in \mathbb{N}}$ and let $\bz = [z_1, z_2, \dots]$ be the realization of the observability vector $\bZ$.
    \begin{enumerate}
        \item  We say that $\bs_t$ is an \emph{anchor location} if $\bs_{t-1} \neq \bs_t$ or $\bs_t \neq \bs_{t+1}$.
        \item We say that $\bs_t$ is an \emph{observed anchor location} if $z_t = 1$ and $\mathbbm{1}(\bs_{t-1} \neq \bs_t) z_{t-1} = 1$ or $\mathbbm{1}(\bs_{t+1} \neq \bs_t) z_{t+1} = 1$.
        \item We say that two observed anchor locations $\bs_t, \bs_r$ with $t<r$ are \emph{consecutive} if all locations between them are observed, i.e. if $z_{t+1} \cdot z_{t+2}\dots z_{r-2} \cdot z_{r-1} = 1$.
    \end{enumerate}
\end{definition}

Intuitively, and as we should explain in more detail below, anchor locations are those locations which allow us to identify increments. We are now ready to write the

An observed anchor location is a location about which we know that it is anchor based on the observations. Contrast that with the situation in which $\bs_{t-1} = \bs_t \neq \bs_{t+1}$ and $z_{t-1}z_t(1 - z_{t+1}) = 1$, i.e. we don't observe $z_{t+1}$. In such case $\bs_t$ would be an observed location (because $z_t = 1$ and it would be an anchor location because $\bs_{t-1}\neq\bs_t$ but it would not be an observed anchor location.

\begin{fact}
$L^S_k$, the original location of the increment $\bM_k$, is an anchor location with probability 1.
\end{fact}
\begin{proof}
If $\bM_k$ is a pause then $\bM_{k-1}$ is a flight. Therefore by Assumption \ref{ass:continuity-in-space-flights} $L^S_{k} \neq L^S_{k-1}$. Moreover, by Assumption flights last only one unit of time which means that $\mathcal{S}_{L^T_k} = L^S_{k} \neq L^S{k-1} = \mathcal{S}_{L^T_{k-1}}  = \mathcal{S}_{L^T_k - 1}$. Therefore $\mathcal{S}_{L^T_k}$ is an anchor location. 
If $\bM_k$ is a flight then $\mathcal{S}_{L^T_{k}} = L^S_k \neq L^S_{k-1} = \mathcal{S}_{L^T_k + 1}$. Therefore $\bs_{l^T_k}$, the realization of $\mathcal{S}_{L^T_k}$, is an anchor location.
All (in)equalities hold almost surely.
\end{proof}

\begin{fact}
Every anchor location is the original location for the realization of some increment.
\end{fact}
\begin{proof}
Let $(\mathcal{m}_1, \mathcal{m}_2, \dots, \mathcal{m}_K)$ to be the realization of motion $\mathcal{M}$ and $\{\bs\}_{t \in \mathbb{N}}$ be the realization of its corresponding trajectory. Consider an anchor location $\bs_{t_k}$. From Definition \ref{def:anchor-locations} we know that either (1) $\bs_{t_k - 1} \neq \bs_{t_k}$ or (2) $\bs_{t_k + 1} \neq \bs_{t_k}$. In the case (1) this means that for some $l$ we have a flight $m_l = (t_k - 1, \bs_{t_k - 1}, 1, \bs_{t_k} - \bs_{t_k - 1}, f)$. Therefore $\bs_{t_k}=L^S_{l+1}$. In the case (2) we know that for some $l'$ we similarly have $\mathcal{m}_{l'} = (t_k, \bs_{t_k}, 1, \bs_{t_k} - \bs_{t_k +1}, f)$. This ends the proof.
\end{proof}

The consequence of these two results is 
\begin{fact}\label{fct:bijection}
The two fact above imply that there exists a bijection between the set of all anchor locations and realization of increments.
\end{fact}
Thus we will use $\mathcal{A} = (\bs_{t_1}, \bs_{t_2}, \dots, \bs_{t_k}, \dots, \bs_{t_K})$ to denote the set of all anchor locations. We can now present the

\begin{proof}[Proof of Proposition \ref{prop:unique-trajectory}]
By Fact \ref{fct:bijection} shown above, there exists a 1-1 mapping between each sequence of anchor locations and realization of a trajectory because the locations which are not anchor locations can be inferred. In particular a location $\bs_t$ which is known to not be an anchor location has to be the same as the most recent anchor location $\bs_r, r<t$. This also means that every set of anchor locations uniquely identifies the realization of a motion.
\end{proof}

\begin{proposition}\label{prop:obs-inc-obs-anchors}
    If the only information available is the realization of the motion trajectory $\bs_1, \bs_2, \dots$, then the value $m_k$ of increment $\bM_k$ is observed if and only if (1) $R_k=p$ while $\bs_{t_k}, \bs_{t_{k+1}}$ and $\bs_{t_{k-1}}$ are two pairs of consecutive observed anchor locations or if (2) $R_k=f$ and $\bs_{t_k}$ and $\bs_{t_{k+1}}$ are a pair of consecutive anchor locations.
\end{proposition}
\begin{proof}
The ``if" part can be shown by constructing a function $f$ which maps anchor locations to $m_k$. This function can be written as
\begin{multline*}
f(\bs_{t_{k-1}}, \bs_{t_k}, \bs_{t_{k+1}}) = 
\left( t_k, 
\bs_{t_k}, \right. \\
t_{k+1} - t_k, \;
\mathbbm{1}(\bs_{t_k} = \bs_{t_{k+1}})(\bs_{t_k} - \bs_{t_{k-1}}) + \mathbbm{1}(\bs_{t_k} \neq \bs_{t_{k+1}})(\bs_{t_{k+1}} - \bs_{t_k}), \\
\left. p\mathbbm{1}(\bs_{t_k} = \bs_{t_{k+1}}) + f\mathbbm{1}(\bs_{t_k} \neq \bs_{t_{k+1}}) \right)
\end{multline*}

Regarding the "only if" part, notice that $\bs_{t_k}$ and $\bs_{t_{k+1}}$ need to consecutive be observed anchor locations to determine whether the increment is a flight or a pause. If they are equal, $\bs_{t_k}$ and $\bs_{t_{k-1}}$ need to be consecutive observed anchor location, because then the realization of $\Delta^S_k$ is equal to the realization of $\Delta^S_{k-1}$.
\end{proof}

\begin{proof}[Proof of Proposition \ref{prop:obs-increments}]
Using Proposition $\ref{prop:obs-inc-obs-anchors}$ we need to show that these conditions are equivalent to observing two consecutive anchor locations. 

First let us consider the situation when $\mathcal{m}_k$ is a pause then the anchor location $\bs_{t_k}$ at its beginning is equal to $\bs_{t_{k+1}}$ the anchor location at the beginning of the following flight. Thus both of them are observed anchor locations if and only if we also observe $\bs_{t_{k-1}}$ and $\bs_{t_{k+2}}$. But since $\mathcal{m}_{k-1}$ and $\mathcal{m}_{k+1}$ are flights this is equivalent to observing the realization of their individual trajectories as well as the first element of the trajectory of $\mathcal{m}_{k+2}$. Thus $z_{t_k}z_{t_k-1}z_{t_{k+1}}z_{t_{k+1} + 1}=1$ means that $\bs_{t_k}$ and $\bs_{t_{k+1}}$ are observed anchor locations. In order for them to be consecutive we also need to observe the trajectory of $\mathcal{m}_k$. This proves the first part of the proposition.

Now let $\mathcal{m}_k$ be a flight. In this case $\bs_{t_k} \neq \bs_{t_{k+1}}$. Thus if $z_{t_k}z_{t_{k+1}}=1$ then we know that they are two consecutive anchor locations. Moreover $\bs_{t_k}$ is the realization of $\mathcal{D}(\bM_k)$ while $\bs_{t_{k+1}}$ is the realization of $\floor{\mathcal{D}(\bM_{k+1})}$. This ends the proof.
\end{proof}

\bibliographystyle{apalike}
\bibliography{bibliography}


\end{document}